\documentclass[a4paper,allowtoday,unpublished]{quantumarticle} 
\pdfoutput=1 

\usepackage{dsfont} 	
\usepackage{microtype} 	
\usepackage[numbers,sort&compress]{natbib} 	
\usepackage[utf8]{inputenc} 	
\usepackage[T1]{fontenc} 	
\usepackage[british]{babel} 	
\usepackage[dvipsnames,x11names,table]{xcolor} 	
\usepackage{amsmath,amssymb,amsthm,bm,amsfonts,bbm} 	
\usepackage{mathtools} 	
\usepackage{physics} 	
\usepackage[unicode,colorlinks=true,citecolor=OliveGreen,linkcolor=Purple,urlcolor=Magenta,hypertexnames=false]{hyperref} 	
\usepackage{thmtools}
\usepackage{thm-restate}
\usepackage{caption} 
\usepackage{subcaption} 
\usepackage[capitalize,nameinlink]{cleveref}
\usepackage{cellspace}

\usepackage{enumitem}

\newtheorem{definition}{Definition}
\newtheorem{theorem}{Theorem}[section] 
\newtheorem{corollary}{Corollary}[section]
\newtheorem{lemma}{Lemma}[section]
\theoremstyle{remark}

\usepackage{listings}
\usepackage{adjustbox}

\usepackage{array}
\newcolumntype{D}{>{$\displaystyle}c<{$}}

\usepackage[table,dvipsnames]{xcolor}

\newcommand{\Discr}{\operatorname{Discr}}
\newcommand{\id}{\mathbb{I}}
\DeclarePairedDelimiter{\ceil}{\lceil}{\rceil}
\DeclarePairedDelimiter{\floor}{\lfloor}{\rfloor}

\begin{document}

\title{The most discriminable quantum states in the multicopy regime}

\author{Maria Kvashchuk}
\affiliation{Sorbonne Université, CNRS, LIP6, F-75005 Paris, France}
\affiliation{Institute of Informatics, 2-1-2 Hitotsubashi, Chiyoda-ku, Tokyo, 101-8430, Japan}
\affiliation{Informatics Program, SOKENDAI, 2-1-2 Hitotsubashi, Chiyoda-ku, Tokyo 101-8430, Japan}
\orcid{0009-0007-5415-5664}

\author{Polina Chernyshova}
\affiliation{Sorbonne Université, CNRS, LIP6, F-75005 Paris, France}
\orcid{0009-0000-6843-1083}

\author{Lucas E. A. Porto}
\affiliation{Sorbonne Université, CNRS, LIP6, F-75005 Paris, France}
\orcid{0000-0001-9509-5951}

\author{Ties-A. Ohst}
\affiliation{Department of Physics and Astronomy, Uppsala University, 75120 Uppsala, Sweden}
\affiliation{Nordita, KTH Royal Institute of Technology and Stockholm University, 10691 Stockholm, Sweden}
\orcid{0000-0003-2260-4342}

\author{Lucas B. Vieira}
\affiliation{Department of Computer Science, Technical University of Darmstadt, Darmstadt, Germany}
\orcid{0000-0002-6530-8271}

\author{Marco Túlio Quintino}
\affiliation{Sorbonne Université, CNRS, LIP6, F-75005 Paris, France}
\orcid{0000-0003-1332-3477}

\begin{abstract}
This work investigates which sets of quantum states give rise to the highest achievable success probability in minimum-error state discrimination if multiple copies of the unknown state are given. Specifically, we consider uniformly distributed ensembles of the form $\left\{\frac{1}{N},\rho_i^{\otimes k}\right\}_{i=1}^N$, where $N$ states in dimension $d$ are provided in $k$ identical copies, and derive universal limits in this scenario. For pure state ensembles, we prove that whenever $N$ is large enough to support a state $k$-design, these designs will exactly give rise to the maximally discriminable sets. We further show that when $N$ exceeds the size required for a $k$-design, mixed states can outperform all pure state ensembles. We then recognise that the problem of most discriminable classical states in the multi-copy regime is in one-to-one correspondence to the concept of the multiplicative Bayes capacity of independent uses of classical channels, a concept that emerges naturally in the context of classical information leakage. This connection allows us to completely solve the classical analogue of our problem when $N\geq \binom{d + k - 1}{k}$, and to prove that quantum systems offer a quadratic advantage (in number of copies $k$) over classical ones. Then, we prove that this quantum over classical advantage is strongly reduced when one is restricted to real quantum states, more precisely, when $N \geq k + 1$, pure real qubits only offer a constant advantage over classical bits. Finally, we introduce computational techniques to find sets of most discriminable ensembles and to obtain rigorous universal upper bounds on the maximal success probability for multi-copy state discrimination in cases that are analytically intractable.
\end{abstract}

\maketitle

\tableofcontents

\section{Introduction}

Quantum state discrimination is a fundamental task that quantifies the probability of determining a particular quantum state that is drawn from a given ensemble~\cite{Barnett2009discrimination,Bae2015ReviewQSD}. Among various different interpretations, quantum state discrimination may be viewed as a communication task between a sender, Alice, and a receiver, Bob. Alice selects a quantum state from a mutually agreed-upon set and transmits it to Bob, whose objective is to identify the received state through an optimal quantum measurement. A core aspect of quantum theory relies on the fact that non-orthogonal pure states cannot be perfectly distinguished. In addition to its foundational value, the existence of non-perfectly discriminable states is central in various applications of quantum theory, such as quantum key distribution~\cite{Bennet2020QuantumCrypto, Wolf2021QKDBook}, quantum query complexity~\cite{Chefles2007Unambiguous}, quantum cryptography~\cite{Aaronson2011Copy-Protetion&Money, Molina2012Money}, hypothesis testing~\cite{Hayashi2006HypotesisTestingBook}, and dimension witnessing~\cite{Brunner2013DimensionWitness}. 

For the case of discriminating between two quantum states, the Holevo-Helstrom theorem gives a closed formula for the probability of success~\cite{helstrom1969quantum}. Also in the case of pure quantum states with a group symmetry, also referred to as group covariant or geometrically uniform states, the so-called pretty good measurements, which are specified by a simple formula, are proven to be optimal, see e.g., Ref.~\cite{Zhou2025distinguishability}. When considering arbitrary states, the problem may be phrased in terms of semidefinite programming (SDP)~\cite{Watrous2018Book, 2023SkrzypczykSDP}, but no ``closed-form'' solution is known in this case. 

Given the major role of state discrimination, it is of fundamental importance to understand the limits of this task within quantum theory, and to identify the sets of states that are the most discriminable. This problem was addressed in the single copy regime in Refs.~\cite{Elron2007OptimalEncoding, Heinosaari_2024}, which fully characterise the sets of states that are most discriminable, even in the more general setting of ensembles with non-uniform prior distributions~\cite{Elron2007OptimalEncoding}.

While the discrimination problem with a single copy is of foundational interest, from an information-theoretic perspective, it is often natural and relevant to consider scenarios where one has access to $k$ identical and independent copies of the desired state. This is the situation, for instance, in tasks such as state learning/tomography ~\cite{Anshu2024SurveyLearning, ODonnel2015EfficientTomography}, quantum key distribution~\cite{Wolf2021QKDBook}, and quantum metrology~\cite{Toth2014Metrology}.
Also, a version of multicopy state discrimination curiously appears in the context of instrument discrimination~\cite{Eid2026InstrumentDisc}.

In this work, we introduce the concept of $k$-copy discriminability of a set of quantum states and maximal $k$-copy discriminability for a given system dimension $d$ and number of states $N$. We analyse this concept on arbitrary (possibly mixed) quantum states and on different relevant subsets of states: general (complex) pure quantum states, classical states, and real quantum states. 

In \cref{sec:pure}, we analyse sets of pure states and show that when $N$ is large enough to support a state $k$-design, these designs give rise to the maximally discriminable sets. Then in \cref{sec:mixed}, we show that when $N$ exceeds the size required for a $k$-design, mixed states can outperform all pure state ensembles. Also, we establish a universal upper bound that holds for arbitrary sets of states, and an arbitrary number of copies $k$.

In \cref{sec:classical}, we analyse the analogue classical discrimination problems, in which states are replaced by probability distributions. We recognise that the problem of most discriminable classical states in the multicopy regime is strongly connected to the concept of the multiplicative Bayes capacity of independent uses of classical channels, a concept that emerges naturally in the context of classical information leakage and classical cryptography, and that has been extensively studied~\cite{Alvim2020BookQuantitativeInfoFlow, Alvim2012MeasuringInfoLeakage, Alvim2014NotionsOfLeakage, Espinoza2013Min_entropy, Smith2017bounds, Kopf2010VulnerabilityBounds, Boreale2011AsymptoticInfoLeakage, Boreale2011QuantitativeInfoFlowWithAView, Braun2009QuantitativeNotions}. This connection allows us to completely solve the classical analogue of our problem when $N \geq \binom{d + k - 1}{k}$, and to prove that quantum systems offer a quadratic advantage (in number of copies $k$) over classical ones. 

In \cref{sec:real}, we consider the case of real quantum states, i.e., quantum states defined in $\mathbb{R}^d$ instead of $\mathbb{C}^d$. There, we establish a connection between the most discriminable real quantum states and group spherical $k$-designs. Curiously, we show that the quantum-over-classical advantage in the maximal discriminability of states is strongly reduced when one is restricted to real quantum states. 

In \cref{sec:numerics}, we introduce computational techniques to find sets of most discriminable ensembles and to obtain rigorous universal upper bounds on the maximal success probability for multicopy state discrimination in cases that are analytically intractable. 

We summarise our main results for the $d=2$ case in \cref{app:summary}, we discuss the relationship between most discriminable states and state designs in \cref{sec:relationship}, and present a visual illustration of the most discriminable qubits in \cref{sec:visualisation}.

\section{The $k$-copy discrimination problem}

The quantum state discrimination problem can be mathematically formalised as follows~\cite{Barnett2009discrimination,Bae2015ReviewQSD}. Let $\{p_i, \rho_i\}_{i=1}^N$ be an ensemble of $N$ $d$-dimensional quantum states, that is, $\rho_i \in \mathcal{L}(\mathbb{C}^d)$, $\Tr(\rho_i) = 1$ and $\rho_i \geq 0$ , and $p_i\geq0$ be the probability of preparing state $\rho_i$. The task of quantum state discrimination consists of guessing which state $\rho_i$ was prepared. For this goal, one is allowed to perform a quantum measurement and guess which state was prepared depending on the outcome of such measurement. Quantum measurements are described by a positive operator-valued measure (POVM), a set of linear operators $\{M_i\}_{i=1}^N$ which are positive semidefinite, $M_i \geq 0$ $\forall i$, and add up to the identity, $\sum_{i=1}^{N} M_i = \id$. The maximal probability of discriminating states sampled from an ensemble $\{p_i, \rho_i\}_{i=1}^N$ is then given by
\begin{align}
p_s\Big(\{p_i, \rho_i\}_{i=1}^N \Big) = \max_{\{M_i\}_{i=1}^N} \sum_{i=1}^{N} p_i \Tr(\rho_i M_i),    
\end{align}
where the maximisation is taken over all POVMs $\{M_i\}_{i=1}^N$.

In this work, we focus on the case where the states are sampled from an equally distributed ensemble, that is, $p_i = \frac{1}{N}$ $\forall i\in\{1,\ldots, N\}$. Moreover, we are interested in the multicopy regime, where one receives $k$ identical copies of the state $\rho_i$, that is, $\rho_i^{\otimes k}$, as illustrated in \cref{fig:protocol}.
In this case, we study the maximal probability of discriminating states sampled from an ensemble $\{\frac{1}{N}, \rho_i^{\otimes k}\}_{i=1}^N$. 
This motivates the definition of $k$-copy discriminability, which is the figure of merit that we will use to quantify how discriminable a set of states is in the multicopy regime.

\begin{figure}[h!]\centering
	\includegraphics[width=0.9\linewidth]{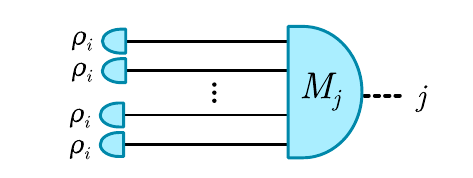}
	\caption{Pictorial illustration of the $k$-copy discrimination scenario considered in this work, where dashed wires represent classical systems and single wires represent quantum systems. With uniform probability, the state $\rho_i^{\otimes k}\in\mathcal{L}({(\mathbb{C}^d)}^{\otimes k})$ is prepared. In order to discriminate among a set of $N$ states, one may perform a joint measurement described by a POVM $\{M_j\}_j$, where $M_j\in \mathcal{L}({(\mathbb{C}^d)}^{\otimes k})$. The discrimination task is successful when $j=i$, and the maximal success probability is presented in \cref{def:discriminability}. }
	\label{fig:protocol}
\end{figure}

\begin{definition}\label{def:discriminability}
The $k$-copy discriminability of a set of $N$ qudit states $\{\rho_i\}_{i=1}^N$ is defined as
     \begin{align}
         \Discr(\{\rho_i\}_{i=1}^N, k) := \max_{\{M_i\}_{i=1}^{N}} \frac{1}{N} \sum_{i=1}^{N} \Tr(\rho_i^{\otimes k} M_i),
    \end{align}
where the maximisation is taken over all positive operator-valued measure (POVM) measurements $\{M_i\}_{i = 1}^N$.
\end{definition}

In other words, the $k$-copy discriminability of a set of states $\{\rho_i\}_{i=1}^N$ is the maximal average probability of correctly guessing a state sampled from an equally distributed ensemble $\{\rho_i^{\otimes k}\}_{i = 1}^N$.
Or, equivalently, it is given by the maximal success probability in the task of minimum-error discrimination of the states $\{\rho_i^{\otimes k}\}_{i = 1}^N$~\cite{Bae2015ReviewQSD}.

We notice that in Refs.~\cite{Heinosaari_2024, Elron2007OptimalEncoding} similar notions have been defined for the single copy case.
In our language, what is defined in Ref.~\cite{Heinosaari_2024} as the \textit{encoding power} of a set of states $\{\rho_i\}_{i = 1}^N$ reads $\text{EP}(\{\rho_i\}_{i=1}^N):=  \Discr(\{\rho_i\}_{i=1}^N, k=1) N $.
Also, using the definition of the encoding power, it holds that $\Discr(\{\rho_i\}_{i=1}^N, k) = \text{EP}(\{\rho_i^{\otimes k}\}_{i=1}^N) \frac{1}{N}$.

In this work, we are interested in studying the maximal value of $k$-copy discriminability that can be reached among all possible sets of states. Therefore, learning the fundamental limit of the state discrimination problem. 

\begin{definition}\label{def:omega}
    The maximal $k$-copy discriminability of $N$ qudit states is defined as
\begin{align}\label{eq:def_omega}
       \Omega^{\textup{q}}(d, N, k) := \max_{\{\rho_i\}_{i=1}^{N}} \max_{\{M_i\}_{i=1}^{N}} \frac{1}{N} \sum_{i=1}^{N} \Tr(\rho_i^{\otimes k} M_i),
\end{align}
where the maximisation is taken over all sets of $d$-dimensional states $\{\rho_i\}_{i=1}^N$ and POVMs $\{M_i\}_{i=1}^N$.
\end{definition}

Moreover, we are also interested in identifying the sets of states that are the most $k$-copy discriminable.

\begin{definition} \label{def:most_discr}
A set of $N$ qudit states $\{\rho_i\}_{i=1}^N$ has maximal $k$-copy discriminability when 
    \begin{align}
        \Discr(\{\rho_i\}_{i=1}^N, k)  =  \Omega^{\textup{q}}(d, N, k).
    \end{align}
\end{definition}

In \cref{def:omega} and \cref{def:most_discr}, the optimisation is performed over all possible quantum states, leading to an ultimate quantum bound on state discrimination in the multicopy regime. While this quantity is meaningful on its own, we are also interested in analysing and comparing this quantity to quantities corresponding to other relevant subsets of quantum states. In \cref{sec:pure}, we introduce $\Omega^{\textup{q}}_\textup{pure}(d, N, k)$, for which maximisation is performed among all sets of $d$-dimensional pure states $\{\ket{\psi_i}\}_{i=1}^N$, and in \cref{sec:classical}, we introduce $\Omega^\textup{cl}(d, N, k)$, for the maximisation over classical states (mathematically formalised by setting states and measurements diagonal in the computational basis). In \cref{sec:real}, we introduce $\Omega^\textup{real}(d, N, k)$, where the optimisation is restricted to states and measurements that make no use of complex numbers, i.e., states defined in $\mathbb{R}^d$ instead of $\mathbb{C}^d$.

\section{Mathematical preliminaries}

In this section we present the definitions of the symmetric subspace, permutation invariance, and state $k$-designs, which are concepts that play a major role in this work.
For an in-depth discussion about the symmetric subspace, we refer to Refs.~\cite{harrow2013churchsymmetricsubspace, 2024QuantMeleHaarMeasureTools}.

\subsection{The symmetric subspace}\label{subsection:sym_k}

\begin{definition}[Symmetric subspace]
    The symmetric subspace of $\left(\mathbb{C}^d\right)^{\otimes k}$ is defined as
    \begin{align}
        \operatorname{Sym}_d^k := \{&\ket{\psi} \in (\mathbb{C}^d)^{\otimes k}: \nonumber\\
        & V_\pi\ket{\psi} = \ket{\psi} \quad \forall \pi \in \mathcal{S}_{k}\},
    \end{align}
    where $\mathcal{S}_{k}$ denotes the symmetric group on $k$ elements, and $V_\pi$ denotes the permutation operator that is associated to $\pi \in \mathcal{S}_k$ via
    \begin{align}
        V_\pi := \sum_{i_1, \ldots, i_k \in [d]}\ket{i_{\pi^{-1}(1)}, \ldots, i_{\pi^{-1}(k)}} \bra{i_1, \ldots, i_k}, \label{eq:permutation_operator_def}
    \end{align}
    where $[d] = \{0, \ldots, d-1\}$.

    We denote by $\mathcal{L}(\operatorname{Sym}_d^k) \subseteq \mathcal{L}({(\mathbb{C}^d)}^{\otimes k})$ the set of linear operators acting on $\left(\mathbb{C}^d\right)^{\otimes k}$ with support on $\operatorname{Sym}_d^k$.
\end{definition}

Simply speaking, a state belongs to the symmetric subspace, i.e., $\ket{\psi} \in \operatorname{Sym}_d^k$ if it is invariant under any permutation of its tensor factors.
The symmetric subspace is also often referred to as the bosonic-symmetric subspace.

As an important result, $k$-fold pure quantum states span the symmetric subspace by Theorem 3 from Ref.~\cite{harrow2013churchsymmetricsubspace} in the sense that
\begin{equation}
    \operatorname{Sym}_d^k = \operatorname{span}\{\ket{\psi}^{\otimes k}: \ket{\psi} \in \mathbb{C}^d\}.
\end{equation}

The projector onto the symmetric subspace can be written as
\begin{equation} \label{eq:symmetric_projection_permutation_form}
   \Pi_\textup{sym}(d,k) = \frac{1}{k!} \sum_{\pi\in S_k} V_\pi.
\end{equation}
Moreover, the projector can be shown to be equal to the average over all $k$-fold pure states, with respect to the Haar measure~\cite{harrow2013churchsymmetricsubspace}
\begin{align}
     \int_\textup{Haar} \ketbra{\psi}^{\otimes k} \mathrm{d}\ket{\psi}  :=& \int_\textup{Haar} \Big(U\ketbra{0} U^\dagger\Big)^{\otimes k} \mathrm{d}U \\
    =&  \frac{\Pi_\textup{sym}(d,k)}{\Tr(\Pi_\textup{sym}(d,k))}.
\end{align}
Finally, the dimension of the symmetric subspace of $\left(\mathbb{C}^d\right)^{\otimes k}$ can, using Eq.~\eqref{eq:permutation_operator_def} and \eqref{eq:symmetric_projection_permutation_form}, be computed as
\begin{align}\label{eq:dim_sym_subspace}
    d_\textup{sym}[k]:=\Tr\Big(\Pi_\textup{sym}(d,k)\Big) = \binom{d + k - 1}{k}.
\end{align}

\subsection{The permutation invariant subspace}

\begin{definition}[Permutation invariant space]
    The permutation invariant space $\operatorname{Perm}_d^k$ is defined as
    \begin{align}
        \operatorname{Perm}_d^k := \{& \rho \in \mathcal{L}(\mathbb{C}^d)^{\otimes k}: \nonumber \\
        & V_\pi \rho V^\dagger_\pi = \rho \quad \forall \pi \in \mathcal{S}_{k}\}.
    \end{align}
    That is, $\operatorname{Perm}_d^k$ is the set of linear operators acting on $(\mathbb{C}^d)^{\otimes k}$ which are invariant under any permutation of the $k$ factors.
\end{definition}

One of the key differences between linear operators that act on the bosonic-symmetric subspace $\mathcal{L}(\operatorname{Sym}_d^k)$ and linear operators from the permutation invariant space $\operatorname{Perm}_d^k$ is that if $\rho \in \mathcal{L}(\operatorname{Sym}_d^k)$ then $\rho = V_\pi \rho V^\dagger_\sigma$ $\forall \pi, \sigma \in \mathcal{S}_k$, and if $\rho \in \operatorname{Perm}_d^k$, it is stable only under applying the same permutation from the left and right sides, i.e., $\pi=\sigma$. A typical example of a state that belongs to the permutation invariant space but not to the symmetric subspace is the maximally mixed state $\rho = \frac{\id}{d^k} \in \operatorname{Perm}_d^k, \notin \mathcal{L}(\operatorname{Sym}_d^k)$.

For any quantum state $\rho \in \mathcal{L}(\mathbb{C}^d)$, the $k$-folded version $\rho^{\otimes k}$ belongs to the permutation invariant space $\rho^{\otimes k} \in \operatorname{Perm}_d^k$. Moreover, it can be shown that~\cite{harrow2013churchsymmetricsubspace}
    \begin{align}
        \operatorname{Perm}_d^k = \operatorname{span}\{\rho^{\otimes k}: \rho \in \mathcal{L}(\mathbb{C}^d)\}.
    \end{align}

\subsection{Quantum state $k$-design}\label{subsection:state_design}

We now present the definition of a weighted quantum state $k$-design~\cite{ambainis2007quantumtdesignstwiseindependence}, which plays a relevant role in this work.
\begin{definition}\label{def:sets_containing_designs}
  We say that a set of $N$ qudit states $\{ \rho_i \}_{i=1}^N$ is a weighted state $k$-design if there exists a probability distribution $\{p_i\}_{i=1}^N$ such that
\begin{align}
    \sum_{i=1}^N p_i \rho_i^{\otimes k} &= \int_\textup{Haar} \ketbra{\psi}^{\otimes k} \mathrm{d}\ket{\psi} \\
    &=  \frac{\Pi_\textup{sym}(d,k)}{\Tr(\Pi_\textup{sym}(d,k))},
\end{align}
where $\forall i \: p_i \geq 0$ is weight of a corresponding state, and when $p_i>0$, $\rho_i$ is a pure state, i.e., $\rho_i=\ketbra{\psi_i}$ for some normalised vector $\ket{\psi_i}\in\mathbb{C}^d$.
\end{definition}

The equally weighted or unweighted state $k$-design is a more restricted object, where $\forall p_i = \frac{1}{N}$. It is known that, for any given dimension $d$ and number of copies $k$, there exists a number of states $N$ such that $\{\ket{\psi_i}\}_{i=1}^N$ contains a $k$-design~\cite{ambainis2007quantumtdesignstwiseindependence, Hayashi2005StateEstimation,Seymour1984averaging}. Let us denote by $N'(d,k)$ the smallest $N$ such that a qudit state $k$-design exists. There exists known lower bound for it ~\cite{weighted_designs_roy_scott, ambainis2007quantumtdesignstwiseindependence, Hayashi2005StateEstimation,Bondarenko2010asymptotic} and it reads as
\begin{align}\label{eq:lb_N_design_existence}
    N'(d, k) \geq \binom{d + \lceil \frac{k}{2} \rceil - 1}{\lceil \frac{k}{2} \rceil} \binom{d + \lfloor \frac{k}{2} \rfloor - 1}{\lfloor \frac{k}{2} \rfloor}.
\end{align}
The equality is achieved only if the design has a uniform weight, i.e., an equally weighted design. 

Before proceeding, it is worth noticing that if a set of states contains a $k$-design, then it also contains a $k'<k$-design. Also, using the Bloch sphere representation of a qubit, for dimension $d=2$, a (unweighted) state $k$-design coincides with a $3$-dimensional spherical design~\cite{Seymour1984averaging}, a concept that is very well studied. For instance, one may consult the website from Ref.~\cite{spehical_website} for various explicit descriptions of $3$-dimensional spherical $k$-designs using minimal or small number of vectors $N$.
We also notice that a uniform state $1$-design coincides with the concept of a tight frame and corresponds to a rank-1 POVM.

For various practical reasons, it is also useful to consider the concept of $\epsilon$-approximate state $k$-design~\cite{ambainis2007quantumtdesignstwiseindependence}.
\begin{definition}\label{def:epsilon-design}
        An ensemble of states $\{p_i,\ketbra{\psi_i}\}_{i=1}^N$ is an $\epsilon$-approximate state $k$-design, if the operator inequalities
    \begin{align}
    &(1-\epsilon) \frac{\Pi_\textup{sym}(d,k)}{\Tr(\Pi_\textup{sym}(d,k))} \leq \sum_{i=1}^{N} p_i \ketbra{\psi_i}^{\otimes k}, \\
    & \sum_{i=1}^{N} p_i \ketbra{\psi_i}^{\otimes k} \leq (1+\epsilon) \frac{\Pi_\textup{sym}(d,k)}{\Tr(\Pi_\textup{sym}(d,k))},
    \end{align}
    are satisfied.
\end{definition}
Ref.~\cite{ambainis2007quantumtdesignstwiseindependence} shows that for any constant $k$ and for every $d>2k$ there exists an $\epsilon$-approximate $k$-design with $\epsilon=O(d^{-1/3})$ consisting of $N=O(d^{3k})$ states. Moreover the states in this approximate design can be efficiently implemented, for various notions of efficiency, including quantum circuits with low depth~\cite{Cui2025design}.

Finally, we present the concept of a real state design.
\begin{definition}\label{def:real_k_design}
An ensemble of $N$ real qudit states $\{p_i,\rho_i \}_{i=1}^N$ is a real state $k$-design when $\rho_i$ is pure, i.e., $\rho_i=\ketbra{\psi_i}$, for some normalised vector $\ket{\psi_i}\in\mathbb{R}^d$, and\footnote{A linear operator $O \in \mathcal{L}(\mathbb{R}^d)$ is in $SO(d)$ if $O O^T = \id$ and $\det O=1$, or, equivalently, if $O$ is a unitary operator with only real numbers in the computational basis and determinant $1$.}
\begin{align}
    \sum_{i=1}^N p_i \rho_i^{\otimes k}
    &= \int_{\ket{\psi}\in\mathbb{R}^d,\textup{Haar}} \ketbra{\psi}^{\otimes k} \mathrm{d}\ket{\psi} \\
    &= \int_{O\in SO(d),\textup{Haar}} \left(O\ketbra{0} O^\dagger\right)^{\otimes k} \mathrm{d}O. 
\end{align}

If $p_i=\frac{1}{N}$, the set of real states $\{\ketbra{\psi_i}\}_{i=1}^N$ is group covariant \footnote{A set of states $\{\rho_i\}_i$ is group covariant if for every $i$, we can write $\rho_i=U_i \rho_1 U_i^\dagger$, and the set of unitaries $\{U_i\}_i$ form a group up to a global phase, i.e., they constitute a projective unitary representation of a group.}, and $\{p_i,\ketbra{\psi_i} \}_{i=1}^N$ is a real state $k$-design  we say that the  $\{\ketbra{\psi_i}\}_{i=1}^N$  is a real state group $k$-design.
\end{definition}
We note that, for dimension $d=2$, a real state design coincides with a $2$-dimensional spherical design, i.e., a circle design.
Also, differently from the complex case, the average over $k$ copies of pure real qudits is not a projector, but a linear operator with more complicated structure~\cite{harrow2013churchsymmetricsubspace,Nemoz2025real}. 

\subsection{The big O notation for asymptotic analysis}
In the sections on asymptotic analysis of the maximal $k$-copy discriminability we will be using the big O notation~\cite{Arora2009computational}. Let $f:D \to \mathbb{C}$, the function to be estimated, be complex valued function defined on domain $D$, and $g:\mathbb{N}\to \mathbb{R}_+$ be function from natural numbers to positive real numbers. We define
\begin{align}
    f(n)\sim g(n) \Rightarrow& \lim_{n\rightarrow \infty} \frac{f(n)}{g(n)} = 1, \hfill ``f = g\text{''} \\
    f(n)=O(g(n)) \Rightarrow& \sup_{n \in D} \frac{\lvert f(n) \rvert}{g(n)} < \infty, \hfill ``f \lesssim g\text{''}\\
    f(n)=\Omega'(g(n)) \Rightarrow& \ g(n) = O(f(n)), \hfill ``g \lesssim f\text{''}  \\
    f(n)=\Theta (g(n)) \Rightarrow& \nonumber \\
    0 < \inf_{n \in D} \frac{\lvert f(n) \rvert}{g(n)} \leq& \sup_{n \in D} \frac{\lvert f(n) \rvert}{g(n)} < \infty , \hfill ``f \propto g\text{''}.
\end{align}

In other words, $f(n)=O(g(n))$ when there exists a constant $K\in\mathbb{R}_+$ such that $\lvert f(n) \rvert \leq K g(n)$ holds for all large enough $n$. That is, up to a multiplicative constant, $g$ is an upper bound for $f$, what is analogous to $\Omega'$ and $\Theta$. Lastly, $f(n)\sim g(n)$  holds when $f$ and $g$ have exactly the same asymptotic behaviour.

In some of our proofs, we also make use of the small o notation,
\begin{align}
    f(n)=o(g(n)) \Rightarrow& \lim_{n\rightarrow \infty} \frac{\lvert f(n) \rvert}{g(n)} = 0, \hfill.
\end{align}
In other words, $f(n)=o(g(n))$ when for every $\epsilon >0$, there exists an $n_0\in\mathbb{N}$ such that $n\geq n_0$ implies $\lvert f(n) \rvert \leq \epsilon g(n)$. That is, $f$ is dominated by $g$ asymptotically for any constant $\epsilon$.

In our work we analyse a function of many integer parameters $\Omega(d,N,k)$.
For a fixed dimension $d$ and a large enough set cardinality $N$, we study the function behaviour with the number of copies $k$ going to infinity, which corresponds to a function $f:\mathbb{N}\to [ 0,1 ]$ in the definition of Big-O notation.
This allows us to omit the absolute value signs for $f$ and to use $\lim_{n\rightarrow \infty}$ instead of $\sup_{n\in D}$ in the definition of $O(\cdot)$.

\section{Single-copy most discriminable states}
Before analysing the general $k$-copy problem, we start our discussion with the simplest case of $k=1$, i.e., the case of single-copy state discrimination.
In this regime, the problem of finding the maximal discriminability achievable with quantum systems has been thoroughly explored in the literature~\cite{Elron2007OptimalEncoding, Heinosaari_2024}, and the sets of maximally discriminable states are known to admit a particularly simple characterisation~\cite{Elron2007OptimalEncoding}.
In this section, we review these results with focus on the case of uniform prior, although we note that in Ref.~\cite{Elron2007OptimalEncoding} the authors fully characterise the sets of $1$-copy maximally discriminable states with respect to any arbitrary prior.

To begin with, when the number of candidate states $N$ is at most equal to the dimension $d$ of the underlying quantum system, i.e., $N \leq d$, one natural strategy that achieves perfect discriminability amounts to choosing the states $\{\rho_i\}_{i = 1}^N$ to be orthogonal pure states.
More precisely, by setting $\rho_i:=\ketbra{i}$ for $i \in \{1, \ldots, N\}$, and constructing the measurement $\{M_i\}_{i = 1}^N$ by $M_i = \ketbra{i} + \frac{1}{N}(\id - \sum_{j=1}^{N} \ketbra{j})$, one can show that $\frac{1}{N} \sum_{i=1}^{N} \Tr(\rho_i M_i) =1$. 

Analogously, this configuration also inspires a possible strategy for the case $N > d$, although perfect discrimination is not achieved.
In this instance, setting $\rho_i:=\ketbra{i}$ for $i \in \{1, \ldots, d\}$ and for $i > d$ taking $\rho_i$ to be any $d$-dimensional state, one can verify that the measurement $\{M_i\}_{i = 1}^N$ defined by $M_i = \ketbra{i}$ for $i \leq d$ and $M_i = 0$ for $i > d$ achieves $\frac{1}{N} \sum_{i=1}^{N} \Tr(\rho_i M_i) = \frac{d}{N}$.

For the former case, it is clear that the strategy we described is optimal, since it reaches perfect discriminability.
However, even for the latter case the presented strategy turns out to be optimal~\cite{Elron2007OptimalEncoding}.
That is, when $N > d$, there is no set of $N$ $d$-dimensional quantum states that can be discriminated with average success probability greater than $\frac{d}{N}$.
Since both strategies can be done using only pure quantum states and, more dramatically, even only using classical ones, this means that
\begin{align}
    \Omega^{\textup{q}}(d, N, k=1)
    &= \Omega^{\textup{q}}_{\textup{pure}}(d, N, k=1) \nonumber \\
    &= \Omega^\textup{cl}(d, N, k=1)\nonumber \\
    &= \min\left(1,\frac{d}{N}\right).
\end{align}
A proof of this fact can be found in Refs.~\cite{Elron2007OptimalEncoding, Heinosaari_2024} (see also \cref{thm:1-copy_max_discriminable} below).

Although these strategies are perhaps the most intuitive way to achieve maximal $1$-copy discriminability with quantum theory, it is interesting to notice that in general there exists several other sets of maximally discriminable states.
For example, in the case $N \geq d$, consider any set of pure states $\{\ket{\psi}_i\}_{i=1}^N$ which contains the maximally mixed states in its convex hull, i.e., for which there exists a probability distribution $\{p_i\}_{i = 1}^N$ such that
\begin{align}\label{eq:id_in_conv_hull}
    \sum_{i=1}^N p_i\ketbra{\psi_i} = \frac{\id}{d}. 
\end{align}
By setting $M_i=dp_i\ketbra{\psi_i}$, the condition above guarantees that the measurement $\{M_i\}_{i = 1}^N$ is properly normalised, and one can verify that $\frac{1}{N} \sum_{i=1}^{N} \Tr(\rho_i M_i) = \frac{d}{N}$.

In fact, when the number of candidate states $N$ is larger than or equal to the dimension of the quantum states $d$, i.e., $N \geq d$, all sets of $1$-copy maximally discriminable states are completely characterised by \cref{eq:id_in_conv_hull}, or, more precisely, by \cref{def:sets_containing_designs}. 
This is essentially the content of \cref{thm:1-copy_max_discriminable} below, which was first proven in Ref.~\cite{Elron2007OptimalEncoding}.
Here, we provide an alternative proof for it, which can also be used to conclude that one cannot achieve a discriminability higher than $\textup{min}(1, \frac{d}{N})$.
\begin{theorem}\label{thm:1-copy_max_discriminable}
    When $N\geq d$, a set of qudit states $\{\rho_i\}_{i=1}^N$ has maximal $1$-copy discriminability if and only if it forms a weighted state $1$-design. 
\end{theorem}
\begin{proof}
First, assume that $\{\rho_i\}_{i=1}^N$ forms a state $1$-design, i.e., that there exists a probability distribution $\{p_i\}_{i =1 }^N$ such that $\sum_{i = 1}^N p_i \rho_i = \frac{\id}{d}$.
Defining $M_i:= d p_i  \rho_i$, one can verify that
\begin{align}
       \frac{1}{N} \sum_{i=1}^{N} \Tr(\rho_i M_i) = \frac{d}{N},
\end{align}
and thus the set $\{\rho_i\}_{i=1}^N$ is maximally discriminable.

Now, we will show that if $\{\rho_i\}_{i=1}^N$ has maximal $1$-copy discriminability, then it forms a weighted state $1$-design.
To do so, notice that for any set of measurements $\{M_i\}_{i = 1}^N$ it holds that
\begin{align}\label{eq:Ties_1-copy}
    \sum_{i=1}^{N} \Tr(\rho_i M_i) & \leq \sum_{i=1}^N \sqrt{\Tr(\rho_i^2)\Tr(M_i^2)} \\
    &\leq \sum_{i=1}^N \sqrt{\Tr(\rho_i)^2\Tr(M_i)^2} \\
    &= \sum_{i=1}^N \Tr(M_i) = d,
\end{align}
where the first inequality is a form of the Cauchy-Schwarz inequality for the Hilbert-Schmidt inner product, and the second inequality follows from the identity $\Tr(A^2) \leq \Tr(A)^2$, valid for any positive semidefinite operator $A$.

A set of measurements $\{M_i\}_{i=1}^N$ that optimally discriminates the states $\{\rho_i\}_{i=1}^N$, i.e., that achieves $\sum_{i=1}^N \Tr(\rho_i M_i) = d$, must therefore saturate both inequalities above.
In the first place, to achieve equality in the Cauchy-Schwarz inequality, if $M_i \neq 0$ we must have $M_i \propto \rho_i$.
Additionally, to saturate the identity $\Tr(M_i^2) \leq \Tr(M_i)^2$, a nonzero $M_i$ must be rank-1. An analogous argument ensures that, when $M_i\neq 0$, the states $\rho_i$ must have rank-1, hence, they are pure. 
These two facts taken together imply that $\{\rho_i\}_{i=1}^N$ is a 1-copy state design.
\end{proof}

Now, coming back to the cases where the number of candidate states is smaller than the dimension of the underlying quantum system $N\leq d$, maximal discriminability amounts to perfect discriminability.\footnote{Notice that the case $N = d$ falls within the scope of both \cref{thm:1-copy_max_discriminable} and \cref{thm:orthogonal_supp}}
In this sense, a set of states is maximally discriminable in this regime if and only if the states are pairwise orthogonal.
This fact is well known and has been mentioned several times in literature~\cite{Watrous2018Book, Bae2015ReviewQSD, Elron2007OptimalEncoding}.
For completeness, we include a proof below.

\begin{theorem}\label{thm:orthogonal_supp}
        When $N\leq d$, a set of qudit states $\{\rho_i\}_{i=1}^N$ has maximal $1$-copy discriminability if and only if $\rho_i\rho_j=0$ for $i \ne j$. 
\end{theorem}
\begin{proof}
        If the states in the set $\{\rho_i\}_{i=1}^{N}$ have orthogonal support, each state can be written as $\rho_i=\sum_k p_{k}^{(i)} \ketbra*{e_{k}^{(i)}}$, where $p_k^{(i)} > 0$ and $\braket*{e_{k}^{(i)}}{e_{k'}^{(j)}}=\delta_{ij} \cdot \delta_{kk'}$. 
        Hence, we can set $M_i:=\sum_k \ketbra*{e_{k}^{(i)}} + \frac{1}{N}(\id - \sum_{i=1}^N \sum_{k}\ketbra*{e_{k}^{(i)}})$ and check that $\frac{1}{N} \sum_{i=1}^{N} \Tr(\rho_i M_i) = 1$. 
        
        To prove the other direction, assume that the set $\{\rho_i\}_{i=1}^{N}$ can be perfectly discriminated, i.e., for some $\{M_i\}_{i=1}^{N}$, we have $\Tr(\rho_i M_j) = \delta_{ij}$.
        This means that for every $i \neq j$, the support of $\rho_i$ must be orthogonal to the support of $M_j$, that is,
        $\textup{supp}(\rho_i) \perp \textup{supp}(M_j)$.
        Additionally, writing $\rho_i=\sum_k p_{k}^{(i)} \ketbra*{e_{k}^{(i)}}$, where $p_k^{(i)} > 0$, it follows that $\bra*{e_{k}^{(i)}} M_i \ket*{e_{k}^{(i)}} = 1$ for all $k$.
        Since $M_i$ cannot have eigenvalues larger than one, this implies that $M_i \ket*{e_{k}^{(i)}} = \ket*{e_{k}^{(i)}}$, which means that the support of $\rho_i$ is included in the support of $M_i$, i.e., $\textup{supp}(\rho_i) \subset \textup{supp}(M_i)$.
        This can then be used to conclude that for $i \neq j$, $\rho_i$ and $\rho_j$ have orthogonal supports, which leads to $\rho_i \rho_j = 0$.
\end{proof}

In summary, the maximal discriminability of quantum states is very well understood in the single-copy regime.
In these scenarios, quantum systems do not provide an advantage over classical ones, and sets of states which are maximally discriminable are well characterised, either in terms of \cref{def:sets_containing_designs} (when $N \geq d$), or in terms of their mutual orthogonality (when $N \leq d$).
In the following sections, we will observe that the problem becomes much more difficult in the multicopy regime.

In the following, we analyse the discriminability of sets of states within the multicopy regime. When more copies of the states are available, we may use this extra resource to improve the probability of successfully discriminating a set of states. And, for any fixed set of different states and large number of copies $k$, we have that $\textup{Discr}(\{\rho_i\}_{i=1}^N, k\to \infty) \to 1$~\cite{2006harrrow_howmanycopies}.

\section{The $k$-copy most discriminable pure states} \label{sec:pure}

Here, we analyse $k$-copy discriminability of sets of $N$ pure $d$-dimensional quantum states. Analogously to \cref{def:omega}, we now introduce the function $\Omega^{\textup{q}}_\textup{pure}(d, N, k)$.
\begin{definition}\label{def:omega-pure}
    The maximal $k$-copy discriminability of $N$ pure qudit states is defined as
    \begin{align}\label{eq:def_omega_pure}
           &\Omega^{\textup{q}}_\textup{pure}(d, N, k) := \nonumber \\
           &\quad \max_{\{\ket{\psi_i}\}_{i=1}^{N}} \max_{\{M_i\}_{i=1}^{N}} \frac{1}{N} \sum_{i=1}^{N} \Tr(\ketbra{\psi_i}^{\otimes k} M_i),
    \end{align}
    where the maximisation is taken over all sets of $d$-dimensional pure states $\{\ket{\psi_i}\}_{i=1}^N$ and POVMs $\{M_i\}_{i=1}^N$.
\end{definition}

As was mentioned in \cref{subsection:sym_k}, $k$-fold versions of pure quantum states span the symmetric subspace $\operatorname{Sym}_d^k$, which implies that the discrimination task is now defined on that space. From that fact and the idea behind \cref{thm:1-copy_max_discriminable}, one can obtain an intuition of where the upper bound for maximal $k$-copy discriminability of sets of pure quantum states comes from.

\begin{theorem}[Upper bound for pure states]\label{thm:k-copy_upper_bound_pure}
    The maximal $k$-copy discriminability of $N$ pure qudit states respects the upper bound \footnote{If $N \leq \binom{d + k - 1}{k}$, the bound \eqref{eq:pure_quantum_bound} becomes trivial and can be improved to $\Omega^{\textup{q}}_{\textup{pure}} \leq 1$, since $\Omega^{\textup{q}}_{\textup{pure}}$ is defined as a probability.}
\begin{align}
    \label{eq:pure_quantum_bound}
    \Omega^{\textup{q}}_{\textup{pure}}(d, N, k) \leq \frac{\binom{d + k - 1}{k}}{N},
\end{align}
where $\binom{d + k - 1}{k}$ is the dimension of the symmetric subspace of $\left(\mathbb{C}^d\right)^{\otimes k}$. 
\end{theorem}
\begin{proof}
    The proof works in a similar way as the one of Theorem \ref{thm:1-copy_max_discriminable} and again makes use of the Cauchy-Schwarz inequality and the trace inequality $\Tr(A^2) \leq \Tr(A)^2$ for positive operators. For any set of pure states $\{\ket{\psi_i}\}_{i=1}^{N} \subset \mathbb{C}^{d}$ and any POVM $\{M_i\}_{i=1}^{N}$ acting on $k$ copies of $\mathbb{C}^{d}$, using the identity $\Pi_\textup{sym}(d,k) \ket{\psi}^{\otimes k} = \ket{\psi}^{\otimes k}$ one can bound the average success probability as
    \begin{align}\label{eq:Ties_k-copy}
     &\frac{1}{N}\sum_{i=1}^{N}\Tr(\ketbra{\psi_i}^{\otimes k} M_i) \\
     &\quad = \frac{1}{N}\sum_{i=1}^{N}\Tr(\ketbra{\psi_i}^{\otimes k}  M_i \Pi_\textup{sym}(d,k)) \\
     &\quad \leq \frac{1}{N} \sum_{i=1}^N \sqrt{\Tr(\ketbra{\psi_i})^{2k}\Tr(M_i \Pi_\textup{sym}(d,k))^2} \\
     &\quad = \frac{1}{N} \Tr(\Pi_\textup{sym}(d,k)) = \frac{\binom{d + k - 1}{k}}{N}.
    \end{align}
\end{proof}

We now show that the upper bound presented in \cref{thm:k-copy_upper_bound_pure} is attainable by pure states if and only if there exists a corresponding set of states that forms a state $k$-design. 

\begin{theorem}\label{thm:k-copy-optimal-pure-design}
A set of pure qudit states $\{\ket{\psi_i}\}_{i=1}^N$ attains the maximal $k$-copy discriminability, i.e.,
\begin{align}
\Discr(\{\ketbra{\psi_i}\}_{i=1}^N, k) = \frac{\binom{d + k - 1}{k}}{N}
\end{align}
if and only if $\{\ket{\psi_i}\}_{i=1}^N$ forms a weighted state $k$-design.
\end{theorem}
\begin{proof}
    We define $M_i:= \binom{d + k - 1}{k} p_i\ketbra{\psi_i}^{\otimes k} + \frac{\id_{d^k}-\Pi_{\textup{sym}}(d,k)}{N}$. It holds that $M_i\geq0$ and $\sum_{i=1}^N M_i= \id_{d^k}$, so that $\{M_i\}_{i=1}^N$ is a valid POVM. We may now check that: 
\begin{align}
    \frac{1}{N}\sum_{i=1}^N \Tr(\ketbra{\psi_i}^{\otimes k} M_i) =  \frac{\binom{d + k - 1}{k}}{N}.
\end{align}
To prove the other direction, we can use the analogous argument from the proof of \cref{thm:1-copy_max_discriminable}: to reach equality in the Cauchy-Schwarz inequality, the set $\{\ket{\psi_i}\}_{i=1}^N$ must be a state $k$-copy design.
\end{proof}

We observe that \cref{thm:k-copy-optimal-pure-design} establishes a strong relationship between maximal $k$-copy discriminability for pure states and state $k$-designs. This relationship is discussed in \cref{{subsec:pure_and_design}}.

For a given $(d, N, k)$, the state $k$-design does not always exist, which implies that for some sets of parameters, no set of pure states saturates the bound (see the \cref{table:numerics}). Nevertheless, in Ref. ~\cite{ambainis2007quantumtdesignstwiseindependence, Hayashi2005StateEstimation,Seymour1984averaging}, it has been shown that for given $d$ and $k$ with a sufficiently large number of states $N$, the design always exists.

\begin{corollary}\label{coro:omega_pure}
    When $N$ is large enough to contain a $k$-design, it holds that
        \begin{align}
        \Omega^{\textup{q}}_{\textup{pure}}(d, N, k) = \frac{\binom{d + k - 1}{k}}{N}.
    \end{align}
\end{corollary}

We remark that an analogous result to \cref{coro:omega_pure} for unitary operations was obtained in Ref.~\cite{2022BavarescoMuraoQuintino}, where the authors show that the maximal probability for discriminating a set of $N$ $d$-dimensional unitary operations, when $k$ calls are available, is upper bounded by $\frac{\binom{d^2 + k - 1}{k}}{N}$, a bound that is attainable by the unitary group $k$-design.

Once, for some $(d, N', k)$, the weighted design exists, the bound is saturated, and the set of states forming a design has maximal $k$-copy discriminability. Now, if we fix $d$ and $k$ but increase the number of states, i.e., $\forall N \geq N'$, the maximally $k$-copy discriminable set of states is not unique, and one construction consists of a weighted design on $N'$ states and arbitrary $N-N'$ states. Another such construction is possible whenever a weighted design with more non-zero weight states exists, and then it can be used analogously.

As discussed in the preliminaries, in ~\cref{def:epsilon-design}, for any constant $k$, for every $d>2k$ there exists an $\epsilon$-approximate $k$ design for $\epsilon=O(d^{-1/3})$ with $N=O(d^{3k})$ states. 
For a given $(d,N,k)$, the exact design does not necessarily exist, but approximate designs exist for any set of parameters. Moreover, the states in this approximate design can be efficiently generated, and such a set admits high discriminability.

\begin{theorem} \label{thm:epsilon}
    If $\{\ket{\psi_i}\}_{i=1}^N$ is an $\epsilon$-approximate state $k$-design \cite{ambainis2007quantumtdesignstwiseindependence}, i.e., there exists a distribution $p_i$ such that
    \begin{align}
    &(1-\epsilon) \frac{\Pi_\textup{sym}(d,k)}{\Tr(\Pi_\textup{sym}(d,k))} \leq \sum_{i=1}^{N} p_i \ketbra{\psi_i}^{\otimes k}, \\
    & \sum_{i=1}^{N} p_i \ketbra{\psi_i}^{\otimes k} \leq (1+\epsilon) \frac{\Pi_\textup{sym}(d,k)}{\Tr(\Pi_\textup{sym}(d,k))},
    \end{align}
    then, 
    \begin{align}
       \Discr(\{\ketbra{\psi_i}\}_{i=1}^N , k) \geq \frac{(1-\epsilon)}{(1+\epsilon)}\frac{\binom{d + k - 1}{k}}{N} 
    \end{align}
\end{theorem}
\begin{proof}
    We define $M_i:= \frac{1}{(1+\epsilon)}\binom{d + k - 1}{k} p_i\ketbra{\psi_i}^{\otimes k}  + \frac{\id_{d^k}-\Pi_{\textup{sym}}(d,k)}{N}$. It is then straightforward to verify that $M_i\geq0$ and $\sum_{i=1}^N M_i \leq \id_{d^k}$. This means that $\{M_i\}_{i=1}^N$ is a subnormalised POVM but we can append $M_{N+1} = \id_{d^k} - \sum_{i=1}^N M_i$, so $\{M_i\}_{i=1}^{N+1}$ is a valid POVM. Using this measurement, a direct calculation then shows the bound
    \begin{align}
       \Discr(\{\ketbra{\psi_i}\}_{i=1}^N , k) \geq \frac{(1-\epsilon)}{(1+\epsilon)}\frac{\binom{d + k - 1}{k}}{N}. 
    \end{align}
\end{proof}

\subsection{In-depth analysis of the case where $d=2$, $N=3$, and $k\in \mathbb{N}$}\label{subsection:pure_examples}

The first non-trivial scenario is the case $(d, N, k) = (2, 3, 2)$. In this case, no state design exists, i.e., there is no set of states that attains the bound given by \cref{thm:k-copy_upper_bound_pure}. From the numerical results and the algorithm for finding the lower bound of $\Omega^{\textup{q}}_\textup{pure}(d=2, N=3, k=2)$ described in \cref{sec:numerics_lower}, the good candidate for set of qubits with maximal two-copy discriminability is an equilateral triangle on the equator of the Bloch sphere (see \cref{fig:triangle}). In the following theorem we proved that the state ensemble of trine states is indeed the optimal set of states for $(d, N, k) = (2, 3, 2)$\footnote{Moreover, after the first version of this manuscript appeared on arXiv, Hyunho Cha presented us a stronger proof of this result, which shows that even when mixed states are considered, for $(d, N, k) = (2, 3, 2)$, the trine is most discriminable and have that $\Omega^{\textup{q}}(d=2, N=3, k=2) =  \frac{1}{3}\left(\frac{3}{2}+\sqrt{2}\right) \approx 0.9714$. This result is available at Ref.~\cite{cha}.}.

\begin{restatable}{theorem}{thmTrineNew}\label{thm:trine_new}
    For $N = 3$, and $k=2$ it holds that 
    \begin{align}
        \label{eq:232_value}
        &\Omega^{\textup{q}}_\textup{pure}(d=2,N=3,k=2) \\ 
        &\quad =\Omega^\textup{real}_\textup{pure}(d=2,N=3,k=2) \\
        &\quad =\frac{1}{2}+\frac{\sqrt{2}}{3} \approx 0.9714.
    \end{align}
    This maximal value is attained with the state ensemble of trine states that constitutes an equiangular triangle in the XZ-plane of the Bloch sphere.  

    Furthermore, when restricting to real states one has for any $N\geq 3$ that
    \begin{align}
       \Omega^\textup{real}_\textup{pure}(d=2,N \geq 3,k=2)= \frac{1}{N}\left(\frac{3}{2}+\sqrt{2}\right).
    \end{align}
\end{restatable}

The proof of \cref{thm:trine_new} is presented in \cref{app:trine}. The maximal discriminability in this case is achieved by using pretty good measurement (PGM) or square-root measurements \cite{Hausladen01121994}. \cref{thm:trine_new} states results for $\Omega^\textup{real}_\textup{pure}(d, N, k)$ as well. We will have a more in-depth discussion about real states in \cref{sec:real} and generalized results on maximal $k$-copy discriminability for real pure states in \cref{thm:SO(d)} and \cref{thm:SO(2)}.

\begin{figure}[h!]\centering
	\includegraphics[width=0.6\linewidth]{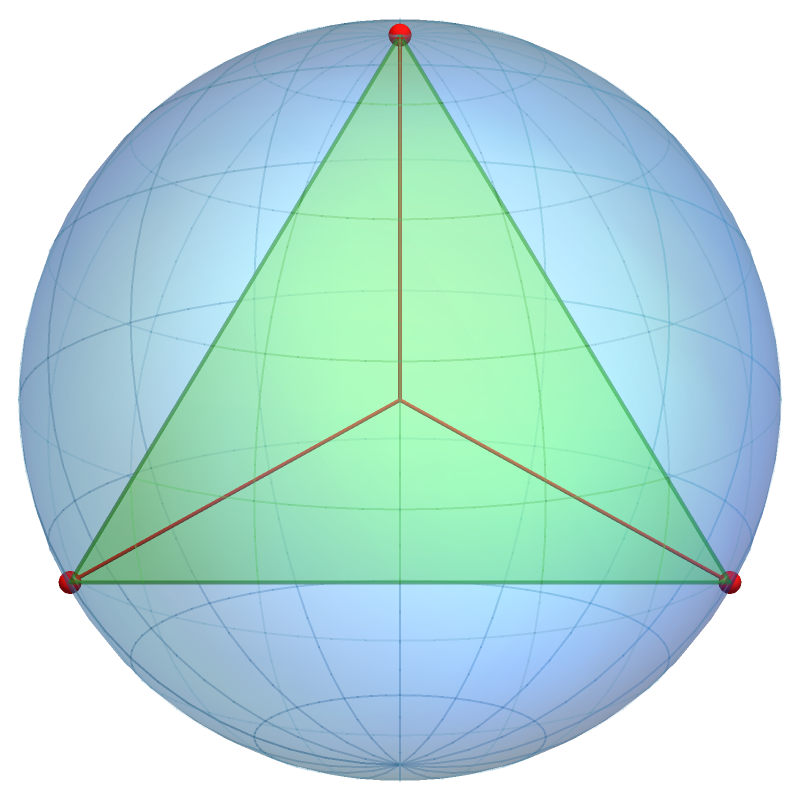}
	\caption{An equilateral triangle in the XZ plane of the Bloch sphere.}
	\label{fig:triangle}
\end{figure} 

We believe that the equilateral triangle set of states also has the maximal discriminability for cases $(d = 2, N = 3, k = 3)$ and $(d = 2, N = 3, k = 4)$, for which we have strong numerical evidence. The exact values of maximal discriminability for all described cases can be found in a \cref{table:numerics} in the appendix.

\subsection{Asymptotic analysis  of $\Omega^{\textup{q}}_{\textup{pure}}(d, N, k)$} \label{subsec:asympt_pure}
Let $N'$ be the smallest number of states such that there exists a $d$-dimensional state $k$-design. Then for any $N \geq N'$, we have that 
\begin{align}
   \Omega^{\textup{q}}_{\textup{pure}}(d, N \geq N', k) &=  \frac{\binom{d + k - 1}{k}}{N}.
\end{align}
Hence, for $d=2$, we have that
\begin{align}
   \Omega^{\textup{q}}_{\textup{pure}}(d=2, N \geq N', k) &=  \frac{k+1}{N},
\end{align}
and, for $d>2$, we may use the inequalities
\begin{align}
   \frac{k^{d-1}}{(d-1)!} \leq \binom{d + k - 1}{k} \leq  \frac{(k+d-1)^{d-1}}{(d-1)!},
\end{align}
to conclude that for sufficiently large $N$, for any dimension $d$, we have that,
\begin{align}
    \Omega^{\textup{q}}_{\textup{pure}}(d, N\geq N', k)  &= \frac{1}{N} \frac{k^{d-1}}{(d-1)!} +  \frac{o(k^{d-1})}{N(d-1)!} \\
    &= \frac{1}{N}\Theta \left(   k^{d-1} \right). \label{eq:pure_asymptotic_theta}
\end{align}

In the limit of large numbers of copies $k$, for fixed $d \geq 2$, we can construct an upper bound for our function. Using $N\geq N'(d,k)\geq C_d k^{2d-2}$, we obtain
\begin{align}
    &\Omega^{\textup{q}}_{\textup{pure}}(d, N, k) \leq \frac{1}{N} \frac{(k+d-1)^{d-1}}{(d-1)!} \\
    &\leq \frac{(k+d-1)^{d-1}}{C_d k^{2d-2}(d-1)!} = \frac{1}{C_d k^{d-1}(d-1)!}(1+\frac{d-1}{k})^{d-1},
\end{align}
For fixed $d$ and sufficiently large $k$, this gives
\begin{align}
    \Omega^{\textup{q}}_{\textup{pure}}(d, N, k)
    &\leq \frac{1}{C_d k^{d-1}(d-1)!}
    \left(1+O\left(\frac{1}{k}\right)\right).
\end{align}
Consequently,
\begin{align}
    \Omega^{\textup{q}}_{\textup{pure}}(d, N, k)
    = O\left(\frac{1}{k^{d-1}}\right).
\end{align}

Here we used that we can bound $N'(d, k)$ as in \cref{eq:lb_N_design_existence}, the result discussed in \cref{subsection:state_design}.
This allows us to see that to saturate the analytical bound in the limit of large numbers of copies $k$, the number of pure quantum states needed to form a state $k$-design scales so quickly, that it makes the optimal discrimination a very hard task, and $\Omega^{\textup{q}}_{\textup{pure}}(d, N, k)$ tends to zero. 

\section{The $k$-copy most discriminable mixed states} \label{sec:mixed}

We now consider the problem on arbitrary, and possibly mixed, quantum states. From the results for the single-copy regime \cref{thm:1-copy_max_discriminable} and \cref{thm:orthogonal_supp}, one could have the intuition that mixed states are less discriminable than the pure ones and that having them in a set will only decrease the discriminability of that set. In fact, for the multicopy regime, the opposite phenomenon can be witnessed. From the numeric results in \cref{table:numerics} and the following theorem, it can be seen that the mixed state strategies are often more beneficial. 

\subsection{Mixed states are more $k$-copy discriminable than pure states}

We start this subsection by presenting a lower bound on the quantity $\Omega^{\textup{q}}(d,N,k)$ that makes use of the value $\Omega^{\textup{q}}_\textup{pure}(d,N-1,k)$.

\begin{theorem}\label{thm:n-1_pure_and_fully_mixed}
For every $d,N,k\in \mathbb{N}$, ($N>1$) it holds that 
\small
\begin{align}
    \Omega^{\textup{q}}(d,N,k) &\geq \frac{N-1}{N} \Omega^{\textup{q}}_\textup{pure}(d,N-1,k) \\
    &\quad {} + \frac{1}{N} \left( 1 - \frac{\binom{d + k - 1}{k} }{d^k}\right).
\end{align}
\normalsize
\end{theorem}
\begin{proof}
    This lower bound can be achieved with the following construction.
    For $i \in \{1, \ldots, N-1\}$, take $\rho_i = \ketbra{\psi_i}$, where $\{\ket{\psi_i}\}_{i = 1}^{N-1}$ is a set of pure states with $\textup{Discr}(\{\ket{\psi_i}\}_{i = 1}^{N-1}, k) = \Omega^{\textup{q}}_{\textup{pure}}(d, N-1, k)$, and for $i = N$, take $\rho_i = \frac{\id}{d}$.
    Then, define the measurement $\{M_i\}_{i = 1}^N$ by taking for $i \in \{1, \ldots, N-1\}$ $M_i = \Pi_{\textup{sym}}(d, k) L_{i} \Pi_{\textup{sym}}(d, k)$, where $\{L_{i}\}_{i=1}^{N-1}$ is any POVM that optimally discriminates $k$-copies of the states $\{\ket{\psi_i}\}_{i = 1}^{N-1}$, and $M_N = \mathbb{I} - \Pi_{\textup{sym}}(d, k)$.
    With this construction, calculation of $\frac{1}{N} \sum_{i=1}^{N} \Tr(\rho_i^{\otimes k} M_i)$ leads to the desired lower bound.
\end{proof}
\cref{thm:n-1_pure_and_fully_mixed} implies the following Corollary that indicates an advantage in the usage of mixed states if the number $N$ of states is large enough.
\begin{corollary} \label{cor:design_and_mixed}
Let $N'$ be the smallest number of states such that there exists a state $k$-design. For any $N>N'$, there exist sets of mixed states that are more discriminable than any set of pure states, that is, $\Omega^{\textup{q}}_\textup{pure}(d,N,k) < \Omega^{\textup{q}}(d,N,k)$, for $d \geq 2$ and $k \geq 2$. 

More precisely,  
 \begin{align}
    \Omega^{\textup{q}}_\textup{pure}(d,N> N',k) = \frac{\binom{d + k - 1}{k} }{N} 
\end{align}
and
\begin{align}
  \Omega^{\textup{q}}(d,N>N',k) \geq \frac{\binom{d + k - 1}{k} }{N} + \frac{1}{N} \left( 1 - \frac{\binom{d + k - 1}{k} }{d^k}\right).
\end{align}
\end{corollary}
In particular, for qubits if there exists a state $k$-design with less than $N$ states, it holds that 
\begin{align}
    \Omega^{\textup{q}}_\textup{pure}(d=2,N\geq N',k) = \frac{k+1}{N} 
\end{align}
and 
\begin{align}
    \Omega^{\textup{q}}(d=2,N>N',k) \geq \frac{1}{N} \left(k + 2 -\frac{k+1}{2^k} \right).
\end{align}

This advantage provided by the usage of mixed states may be understood in the following way.
Allowing for a candidate set to be constructed from $d$-dimensional qudit states $\{\rho_i\}_{i=1}^N$, $\forall i$ $\rho_i \in \mathcal{L}(\mathbb{C}^d)$ implies that the discrimination problem is defined on a permutation-invariant space since $\rho_i^{\otimes k} \in \operatorname{Perm}_d^k$. The permutation-invariant space is bigger than the space $\mathcal{L}(\operatorname{Sym}_d^k)$ of operators supported on the bosonic-symmetric subspace; for example, a maximally mixed state belongs to $\operatorname{Perm}_d^k$ but not to $\mathcal{L}(\operatorname{Sym}_d^k)$, which leads to the fact that the optimal strategy now involves the optimisation over a bigger search space. In the optimal measurements in \cref{thm:n-1_pure_and_fully_mixed}, the larger space is utilised by having a measurement outcome defined as $M_N = \mathbb{I} - \Pi_{\textup{sym}}(d, k)$, i.e., $M_N \perp \mathcal{L}(\operatorname{Sym}_d^k)$, in contrast to the optimal measurements for the pure state case \cref{thm:k-copy_upper_bound_pure} where the subspace orthogonal to the symmetric subspace is not used. All $k$-fold versions of pure states in the candidate set will never click in the $M_N$ measurement element, since $\ketbra{\psi_i}^{\otimes k} \in \mathcal{L}(\operatorname{Sym}_d^k)$. This is where the increase of maximal $k$-copy discriminability comes from.

\subsection{An upper bound on $\Omega^{\textup{q}}(d,N,k)$} 

We now present a fundamental limitation on the discrimination of $k$ copies of arbitrary sets of quantum states that are uniformly distributed. This upper bound holds not only for pure states, but also for mixed ones.

\begin{theorem}[Upper bound for mixed states] \label{thm:k-copy_upper_bound_mixed}
The maximal $k$-copy discriminability of $N$ qudit states respects the upper bound
\begin{align}\label{eq:upper_mixed}
    \Omega^{\textup{q}}(d, N, k) \leq  \frac{\binom{d^2 + k - 1}{k}}{N},
\end{align}
where $\binom{d^2 + k - 1}{k}$ is the dimension of the symmetric subspace of $(\mathbb{C}^{d^2})^{\otimes k}$.
\end{theorem}
\begin{proof}
    Discriminating between $k$ copies of a set of mixed states $\{\rho_i\}_{i = 1}^N$ is never easier than discriminating between $k$ copies of their purifications.
    Since such purifications always exist in dimension $d^2$, the upper bound \cref{eq:upper_mixed} follows from \cref{thm:k-copy_upper_bound_pure}. 
\end{proof}

It is worth noting that, while this universal upper bound is not tight (and likely not attainable), as discussed in \cref{subsec:AsymptoticOmega}, this upper bound is not far from being saturated by pure states when the number of copies $k$ is large and $N$ is large enough to contain a state $k$-design.

\subsection{In-depth analysis of the case where $d=2$, $N=4$, and $k\in \mathbb{N}$}\label{subsection:mixed_examples}
One example of the fact that the discriminability of mixed states is higher than that of pure states can be seen in the scenario $(d, N, k) = (2, 4, 2)$. The state $2$-design for $N=4$ exists, so the maximally discriminable set is given by a tetrahedron \cref{fig:tetrahedron}, and its discriminability attains the bound from \cref{thm:k-copy_upper_bound_pure}: $\Omega^{\textup{q}}_{\textup{pure}}(d = 2, N = 4, k = 2) = 0.75$, which can be achieved using PGM measurements. 
\begin{figure}[h!]\centering
	\includegraphics[width=0.6\linewidth]{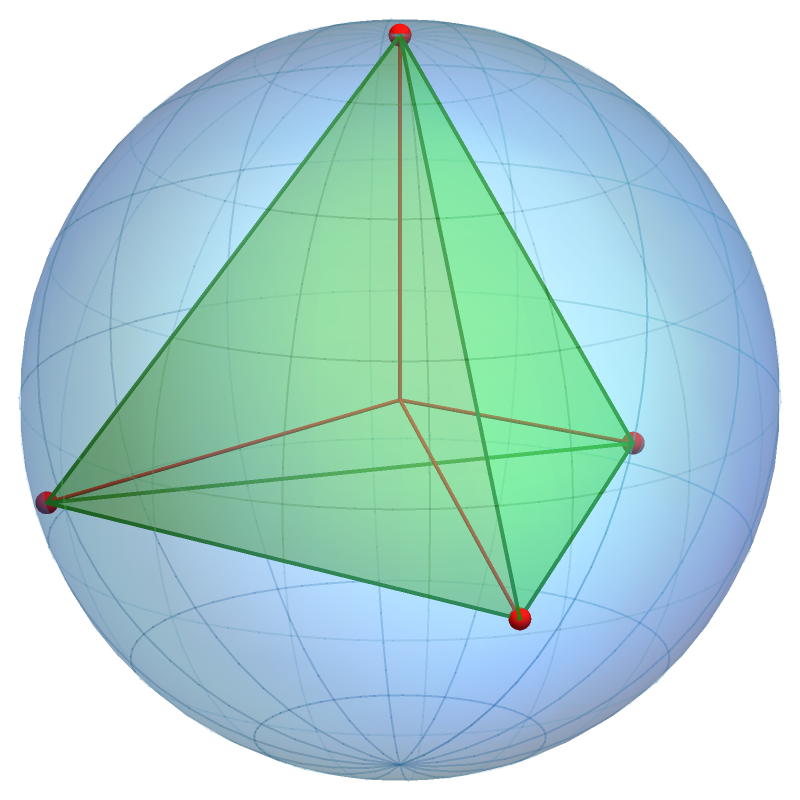}
	\caption{Tetrahedron.}
	\label{fig:tetrahedron}
\end{figure}
However, choosing an equilateral triangle and a fully mixed state as a set of states \cref{fig:triangle_id} will yield the discriminability value $\Omega^{\textup{q}}(d = 2, N = 4, k = 2) = 0.7911$, which is higher than that of the pure state. The optimal measurements for this set admit the construction of measurements for the lower bound in \cref{thm:n-1_pure_and_fully_mixed}.
\begin{figure}[h!]\centering
	\includegraphics[width=0.6\linewidth]{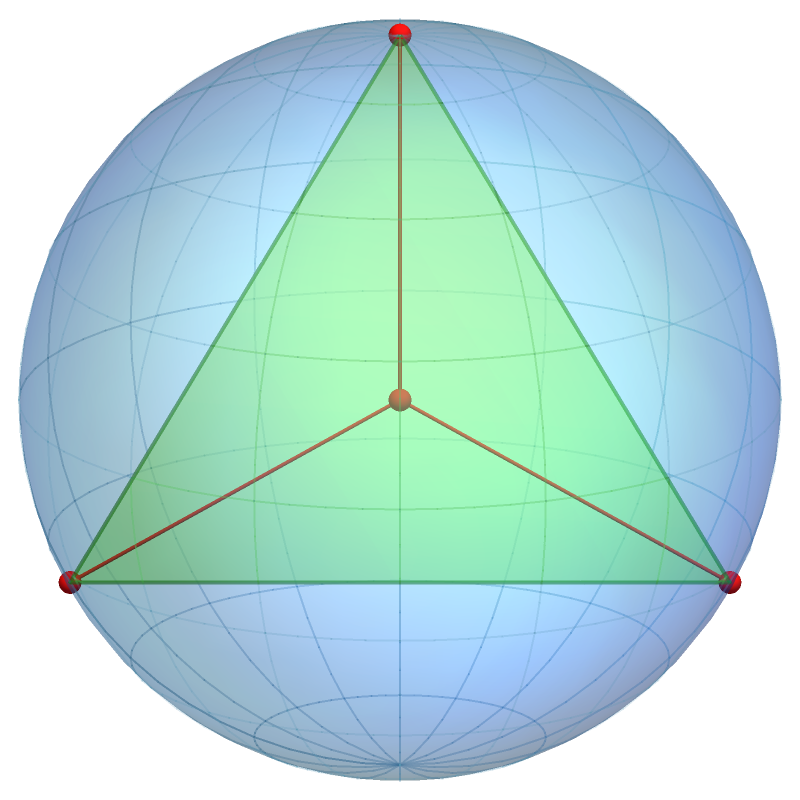}
	\caption{Triangle and fully mixed state.}
	\label{fig:triangle_id}
\end{figure}
If we consider the same two sets of states for the scenarios with higher $k$, for example, in the scenario $(d, N, k) = (2, 4, 3)$, the discriminability of the tetrahedron set of states is higher than that of the triangle and fully mixed sets. The same results are observed for the scenario $(d, N, k) = (2, 4, 4)$. Using analogous numerical techniques as in \cref{subsection:pure_examples}, we have strong numerical evidence that for $(d = 2, N = 4, k = 3)$ and $(d = 2, N = 4, k = 4)$, the tetrahedron set of states is the unique optimal solution.

\subsection{Asymptotic analysis  of $\Omega^{\textup{q}}(d, N, k)$} \label{subsec:AsymptoticOmega}
As discussed in the previous section, for pure states we have 
\begin{align}
  \Omega^{\textup{q}}_{\textup{pure}}(d, N\geq N', k) \leq \frac{1}{N} \frac{(k+d-1)^{d-1}}{(d-1)!}.
\end{align}
From analogous arguments, the upper bound on $\Omega^{\textup{q}}(d,N,k)$ ensures that
\begin{align}
   \Omega^{\textup{q}}(d,N \geq N',k) \leq  \frac{1}{N}\frac{(k+d^2-1)^{d^2-1}}{(d^2-1)!},
\end{align}
from which we conclude that
\begin{align}
   \Omega^{\textup{q}}(d,N \geq N',k) = O\left(\frac{k^{d^2-1}}{N} \right).
\end{align}

The lower bound is constructed in \cref{thm:n-1_pure_and_fully_mixed}, and we can analyse it for the case of large enough $k$ and sufficiently large $N \ge N'(d,k)$:
\begin{align}
   \Omega^{\textup{q}}(d,N \geq N',k) &\geq \frac{1}{N}\left( \frac{k^{d-1}}{(d-1)!} + 1 \right) \\
   &\quad {} - \frac{1}{N} \frac{(k+d-1)^{d-1}}{d^k(d-1)!} ,
\end{align}

which allows us to describe our function
\begin{align}
   \Omega^{\textup{q}}(d,N \geq N',k) = \Omega'\left(\frac{k^{d-1}}{N} \right).
\end{align}

In general, we can say that for a fixed dimension, large $k$, and sufficiently large $N$, we have the scaling $\Omega^{\textup{q}}(d,N \geq N',k) \approx \frac{k^{\text{Poly}(d)}}{N} $.

\section{The $k$-copy most discriminable classical states} \label{sec:classical}

In this section, we investigate the classical version of the problem, restricting the set of possible states only to classical ones.
Interestingly, we establish a direct connection between the $k$-copy discriminability of classical states and the notion of \textit{multiplicative Bayes capacity} of $k$ independent runs of a classical channel, which has been extensively studied in the literature~\cite{Alvim2020BookQuantitativeInfoFlow, Alvim2012MeasuringInfoLeakage, Alvim2014NotionsOfLeakage, Espinoza2013Min_entropy, Smith2017bounds, Kopf2010VulnerabilityBounds, Boreale2011AsymptoticInfoLeakage, Boreale2011QuantitativeInfoFlowWithAView, Braun2009QuantitativeNotions}.
This allows us to prove tight bounds on the maximal discriminability of classical states, which highlights the quantum advantages in this setting.

For convenience of notation, we choose to describe classical states within the quantum formalism.
That is, we describe a classical state by a density matrix which is diagonal in the computational basis, and whose diagonal terms correspond to a probability distribution.
In this sense, a set of $d$-dimensional classical states $\{\rho_i^{\textup{cl}}\}_{i=1}^N$ is fully specified by a conditional probability distribution $\{\{ p(j|i)\}_{j=0}^{d-1}\}_{i=1}^N$ via
\begin{equation}\label{eq:classical_states}
    \rho_i^\textup{cl}=\sum_{j=0}^{d-1} p(j|i) \ketbra{j}.
\end{equation}

Analogously, a classical measurement is described by a POVM whose elements are also diagonal in the computational basis.
Therefore, in a similar way we can fully specify a classical measurement $\{M_i\}_{i=1}^N$ by a conditional probability distribution $\{\{ q(i|j)\}_{i=1}^{N}\}_{j=0}^{d-1}$, according to
\begin{equation}\label{eq:classical_measurement}
    M_i^\textup{cl}=\sum_{j=0}^{d-1} q(i|j) \ketbra{j}.
\end{equation}
Notice, in particular, that a classical measurement can always be viewed as a measurement in the computational basis followed by a classical post-processing.

In the multicopy regime, the states that represent $k$ copies of $\{\rho_i^{\textup{cl}}\}_{i=1}^N$ can be written as
\begin{equation}\label{eq:classical_states_multicopy}
    {\rho_i^\textup{cl}}^{\otimes k}= \sum_{j_1, \ldots, j_k \in [d]} \left(\prod_{l=1}^kp(j_l|i)\right) \ketbra{j_1 \ldots j_k}.
\end{equation}
And a classical measurement acting on $k$ copies of $d$-dimensional states is denoted by
\begin{equation}\label{eq:classical_measurement_multicopy}
    M_i^\textup{cl}=\sum_{j_1, \ldots, j_k \in [d]} q(i|j_1 \ldots j_k) \ketbra{j_1 \ldots j_k},
\end{equation}
where $\{\{ q(i|j_1 \ldots j_k)\}_{i=1}^{N}\}_{j_1, \ldots, j_k \in [d]}$ is a conditional probability distribution.

With this notation, the optimal average probability of correctly discriminating between a $k$ copies of a set of classical states $\{\rho_i^{\textup{cl}}\}_{i=1}^N$ sampled with uniform prior, i.e., the $k$-copy discriminability of $\{\rho_i^{\textup{cl}}\}_{i=1}^N$, has a particularly simple form, given by the following theorem.
\begin{theorem}\label{thm:classical_state_discrimination}
    Let $\{\rho_i^{\textup{cl}}\}_{i=1}^N$ be a set of classical states described by $ \rho_i^\textup{cl}=\sum_{j=0}^{d-1} p(j|i) \ketbra{j} $, where $\{\{ p(j|i)\}_{j=0}^{d-1}\}_{i=1}^N$ is a conditional probability distribution. 
    It holds that
\begin{align}\label{eq:discriminability_classical}
    \Discr(\{\rho_i\}_{i=1}^N,k) = \frac{1}{N} \sum_{j_1, \ldots, j_k \in [d]} \max_{i}  \prod_{l=1}^k p(j_l|i).
\end{align}
When $k=1$, this equation reduces to
\begin{align}
    \Discr(\{\rho_i\}_{i=1}^N,k=1) = \frac{1}{N} \sum_{j=0}^{d-1} \max_i p(j|i).  
\end{align}
\end{theorem}
\begin{proof}
    According to \cref{def:discriminability}, \cref{eq:classical_states_multicopy} and \cref{eq:classical_measurement_multicopy}, the $k$-copy discriminability of $\{\rho_i^{\textup{cl}}\}_{i=1}^N$ can be written as
    \begin{align}
        &\Discr(\{\rho_i\}_{i=1}^N,k) = \nonumber \\
        &\quad \max_{\{q(i \vert j_1 \ldots j_k)\}} \frac{1}{N}\sum_{j_1, \ldots, j_k \in [d]} \sum_{i = 1}^N q(i \vert j_1 \ldots j_k)\prod_{l = 1}^k p(j_l \vert i),
    \end{align}
    where the maximisation is performed among all conditional probability distributions $\{\{ q(i|j_1 \ldots j_k)\}_{i=1}^{N}\}_{j_1, \ldots, j_k \in [d]}$.
    Since for each $j_1, \ldots, j_k$, we can choose the distribution $\{q(i \vert j_1 \ldots j_k)\}_{i=1}^N$ independently of other choices of $j$'s, \cref{eq:discriminability_classical} follows.
\end{proof}

In other words, the theorem above provides a simpler way to find the optimal measurement for the discrimination of any given set of classical states $\{\rho_i^{\textup{cl}}\}_{i=1}^N$, which leads to a more straightforward expression for the discriminability of such states, given by \cref{eq:discriminability_classical}.
As we will shortly discuss, \cref{eq:discriminability_classical} performs a crucial role for the connection between the discriminability of classical states and a well-studied quantity in the literature of classical information theory.

Our goal in establishing this connection is to compute bounds on the maximal $k$-copy discriminability of classical states, defined as follows.
\begin{definition}\label{def:omega-classical}
The maximal $k$-copy discriminability of $N$ classical states is defined as
\begin{align} \label{eq:def_omega_classical}
       &\Omega^\textup{cl}(d, N, k) := \nonumber \\
       &\quad \max_{\{\rho_i^\textup{cl}\}_{i=1}^{N}} \max_{\{M_i^\textup{cl}\}_{i=1}^{N}} \frac{1}{N} \sum_{i=1}^{N} \Tr({\rho^\textup{cl}_i}^{\otimes k} M_i^\textup{cl}),
\end{align}
where the maximisation is taken over all sets of $d$-dimensional diagonal states $\{\rho_i^\textup{cl}\}_{i=1}^N$ and diagonal POVMs $\{M_i^\textup{cl}\}_{i=1}^N$ (see \cref{eq:classical_states} and \cref{eq:classical_measurement}).

Equivalently, in light of \cref{thm:classical_state_discrimination}, the maximal $k$-copy discriminability can be written as
\begin{align}\label{eq:omega-classical-simplified}
    \Omega&^\textup{cl}(d, N, k) = 
       \max_{\{p(j \vert i)\}_{ij}} \frac{1}{N}\sum_{j_1, \ldots, j_k \in [d]} \max_i  \prod_{l=1}^k p(j_l|i).
\end{align}
\end{definition}

\subsection{The $k$-copy pure classical states}
Before discussing the general version of the classical problem, it is interesting to observe that if we restrict ourselves to the maximal $k$-copy discriminability of pure classical states, the problem is greatly simplified.
In our notation, pure classical states are simply states of the computational basis, which can be obtained by taking deterministic probability distributions $\{\{ p(j|i)\}_{j=0}^{d-1}\}_{i=1}^N$ on \cref{eq:classical_states}.

Since there are exactly $d$ pure classical states, the most natural instance of the state discrimination problem is the case $N\leq d$.
In this situation, the set $\{\rho_i^{\textup{cl}}\}_{i=1}^N$ defined by $\rho_i^\textup{cl} = \ketbra{i}, \forall i \in \{1, \ldots, N\}$ is perfectly discriminable, by a measurement in the computational basis.
In the case $N > d$, any set $\{\rho_i^{\textup{cl}}\}_{i=1}^N$ contains at most $d$ different states (which, by the above reasoning, are perfectly discriminable), and therefore the best discriminability one can achieve is $d/N$.

Notice that this reasoning can be applied independently of the considered number of copies $k$, since a projective measurement in the computational basis always perfectly identifies any classical pure state.
That is, in the classical pure case, considering multiple copies does not provide any advantage over the single copy scenario.

Therefore, we conclude that
\begin{equation}\label{eq:max_classical_pure_discriminability}
    \Omega^{\textup{cl}}_{\textup{pure}}(d, N, k) = \min\left(1, \frac{d}{N}\right).
\end{equation}
More formally, \cref{eq:max_classical_pure_discriminability} can also be obtained by restricting the optimisation in \cref{eq:omega-classical-simplified} to deterministic probability distributions, which correspond to pure classical states.

Similarly to the quantum case, randomness (i.e., considering mixed classical states) will also be useful in the multicopy regime.
This is due to the fact that randomness enlarges the set of possible states, and in the limit of arbitrarily many copies, any such states are arbitrarily close to being perfectly discriminable.
In some sense, this is also the reason why mixed quantum states are more $k$-copy discriminable than pure ones, but the enlargement of the state space in the quantum case is somewhat subtle, as it relates to the difference between the symmetric and the permutation invariant subspaces.
In the classical case, this enlargement is more explicit, since there are only finitely many classical pure states, but infinitely many mixed ones.

This discussion is nicely illustrated by a simple example.
Consider the problem of discriminating a set of $N = 3$ classical states of dimension $d = 2$, denoted by $\{b_1, b_2, b_3\}$.
If we restrict the states to be pure, we are thus in the problem of discriminating between three bits.
In this setting, at least two bits will necessarily have the same value, which limits the $1$-copy discriminability to $\Omega^{\textup{cl}}_{\textup{pure}}(d=2, N=3, k=1) = \Discr(\{\ketbra{0},\ketbra{1},\ketbra{1}\},k=1) = \frac{2}{3}$.
It also follows that taking any number $k$ of copies of those bits does not alter their discriminability, since we cannot extract new information from having more copies of the same bit, in accordance with \cref{eq:max_classical_pure_discriminability}.
Alternatively, having access to randomness allows us to discriminate within $\{b_1, b_2, b_3\}$ better, making the set asymptotically perfectly discriminable.
For example, setting $b_1 = \ketbra{0}$, $b_2 = \ketbra{1}$ and $b_3 = \frac{1}{2}\ketbra{0} + \frac{1}{2}\ketbra{1}$, it follows that $\Discr(\{b_1,b_2,b_3\},k) = 1 - \frac{1}{3} \frac{1}{2^{k-1}}$.
Thus, not only $\Discr(\{b_1,b_2,b_3\},k') > \Discr(\{b_1,b_2,b_3\},k)$ for any $k' > k$, but also $\Discr(\{b_1,b_2,b_3\},k) \rightarrow 1$ when $k \rightarrow \infty$.

\subsection{The $k$-copy most discriminable mixed classical states}
We now turn to the general classical version of the problem, of finding the maximal $k$-copy discriminability of arbitrary classical states.
As stated in \cref{def:omega-classical}, $d$-dimensional classical states $\rho_i$ are completely described by a conditional probability distribution $\{\{ p(j|i)\}_{j=0}^{d-1}\}_{i=1}^N$ via the equation $ \rho_i^\textup{cl}=\sum_{j=0}^{d-1} p(j|i) \ketbra{j}$.

From a classical information theory perspective, conditional probability distributions represent classical channels, which may be described by its associated channel matrix\footnote{Traditionally, literature on classical information theory does not make use of Dirac notation. Hence, the authors often just state the $C$ is matrix with matrix coefficients $C_{ij}=p(j|i)$.
}
\begin{align}\label{eq:classical_channel}
    C_{N\to d}=\sum_{i=1}^N \sum_{j=0}^{d-1} p(j|i) \ketbra{j}{i}.
\end{align}
Here, we use the sub-index $N\to d$ to emphasise that $C_{N\to d}$ is a classical channel which maps distributions with $N$ different symbols to distributions with $d$ different symbols. 
Therefore, the conditional probability distributions that describe both sets of classical states and classical channels provides a one-to-one correspondence between these concepts.

With this correspondence, it turns out that the $k$-copy discriminability of a set of classical states $\{\rho^{\textup{cl}}_i\}_{i=1}^N$, as given by \cref{eq:discriminability_classical}, is \textit{precisely} the so-called \textit{multiplicative Bayes capacity} of $k$ independent runs of the associated classical channel $C_{N\to d}$, up to a factor $1/N$.
This quantity, that has been extensively studied in the literature of \textit{quantitative information flow}~\cite{Alvim2020BookQuantitativeInfoFlow, Alvim2012MeasuringInfoLeakage, Alvim2014NotionsOfLeakage, Espinoza2013Min_entropy, Smith2017bounds}, is denoted by $\mathcal{ML}^\times(C^{(k)}_{N\to d})$, where $C_{N\to d}^{(k)}$ denotes the channel that corresponds to $k$ independent runs of $C_{N\to d}$, defined as
\begin{align}
    C_{N\to d}^{(k)}:= \sum_{i=1}^N \sum_{j_1, \ldots, j_k \in [d]} \left(\prod_{l=1}^k p(j_l|i)\right) \ketbra{j_1 j_2 \ldots j_k}{i}.
\end{align}
This connection is more precisely formalised in the following lemma.
\begin{lemma} \label{lemma:ML}
Let $\{\rho_i^\textup{cl}\}_{i=1}^N$ be a set of classical states described by $ \rho_i^\textup{cl}=\sum_{j=0}^{d-1} p(j|i) \ketbra{j}$ and $C_{N\to d}$ be a classical channel with channel matrix
$C_{N\to d}=\sum_{i=1}^N \sum_{j=0}^{d-1} p(j|i) \ketbra{j}{i}$. It holds that
\begin{align}
     \Discr(\{\rho_i^\textup{cl}\}_{i=1}^N,k) = \frac{1}{N} \mathcal{ML}^\times(C^{(k)}_{N\to d}).
\end{align}
\end{lemma}
\begin{proof}
    According to \cite[Proposition 5.1]{Braun2009QuantitativeNotions} (see also~\cite[Theorem 7.2]{Alvim2020BookQuantitativeInfoFlow}), the multiplicative Bayes capacity of $k$ independent runs of the channel $C_{N\to d}$ is given by
    \begin{equation}
        \mathcal{ML}^\times(C^{(k)}_{N\to d}) = \sum_{j_1, \ldots, j_k \in [d]} \max_i  \prod_{l=1}^k p(j_l|i),
    \end{equation}
    which is exactly $N \Discr(\{\rho^{\textup{cl}}_i\}_{i=1}^N,k)$ (see \cref{thm:classical_state_discrimination}). 
\end{proof}

For a definition of the multiplicative Bayes capacity in the context where it is more often studied, we refer the reader to Refs.~\cite{Alvim2012MeasuringInfoLeakage, Alvim2020BookQuantitativeInfoFlow}.
Also, we note that it in the $1$-copy case, it is also equivalent to the notion of \textit{communication value} of a channel, studied, for example, in Refs.~\cite{Chitambar2021CommunicationValue, Doolittle2021ClassicalCommCost} and that also appears in \cite{Cubitt2010ZeroErrorSimulation}.
For discussions on the multiplicative Bayes capacity of multiple independent runs of the same channel, to Refs.~\cite{Espinoza2013Min_entropy, Smith2017bounds}.
Here, we notice that based on \cite[Theorem 3.1] {Smith2017bounds}, this correspondence directly provides an upper bound for the maximal $k$-copy discriminability of classical states, which is tight for large enough $N$.
\begin{theorem} \label{thm:classical_dit}
The maximal $k$-copy discriminability of $N$ qudit classical states respects the upper bound 
\begin{align}\label{eq:upper_classical_dit}
    \Omega^\textup{cl}(d,N,k) \leq& \frac{1}{N} \textup{cap}_d(k),
\end{align}
where $\textup{cap}_d(k)$ is defined in Ref.~\cite{Smith2017bounds} as
\begin{align}\label{eq:cap_d}
\textup{cap}_d(k) := \frac{k!}{k^k} \sum_{\substack{t_1, t_2, \dots, t_d \in \mathbb{N} \\ t_1+t_2+\dots+t_d=k}} \frac{t_1^{t_1} t_2^{t_2} \dots t_d^{t_d}}{t_1! t_2! \dots t_d!}
\end{align}

Moreover, when $N\geq \binom{d + k - 1}{k}$, this upper bound is saturated by a set $\{\rho_i^{\textup{cl}}\}_{i=1}^N$ where each state corresponds to a so-called type (see Refs.~\cite{Harrow2005Thesis, harrow2013churchsymmetricsubspace, Smith2017bounds} for details).
That is, for $N\geq \binom{d + k - 1}{k}$,
\begin{align}\label{eq:upper_classical}
      \Omega^\textup{cl}(d,N,k) = \frac{1}{N} \textup{cap}_d(k).
\end{align}
\end{theorem}
\begin{proof}
It follows from \cref{lemma:ML} that 
\begin{align}
     \Omega^\textup{cl}(d,N,k) = \frac{1}{N}\max_{C^{(k)}_{N\to d}} \mathcal{ML}^\times(C^{(k)}_{N\to d}) .
\end{align}
We now finish the proof by invoking~\cite[Theorem 3.1]{Smith2017bounds}, which ensures that
\begin{align}
    \mathcal{ML}^\times(C^{(k)}_{N\to d}) \leq \textup{cap}_d(k),
\end{align}
bound that is saturated when $N\geq \binom{d + k - 1}{k}$.
\end{proof}

The theorem above thus fully determines the maximal $k$-copy discriminability of classical states for large enough $N$.
For $d$-dimensional states, this quantity matches the $\text{cap}_d(k)$ function, defined in \cref{eq:cap_d}.
It is worth noticing that Ref.~\cite{Smith2017bounds} also presents a recursive formula to evaluate $\text{cap}_d(k)$ for any $k\in\mathbb{N}$ and any $d>2$,
\begin{align}
    \text{cap}_d(k) = \text{cap}_{d-1}(k) + \frac{k}{d-2} \text{cap}_{d-2}(k).
\end{align}

Additionally, \cref{thm:classical_dit} provides examples of maximally $k$-copy discriminable classical states, for large enough $N$.
These are associated with the information-theoretic notion of \textit{types}~\cite{harrow2013churchsymmetricsubspace, Harrow2005Thesis}.
In fact, the proof of~\cite[Theorem 3.1]{Smith2017bounds} essentially follows from properties of types.
Since its main ideas can be already seen in the case $d = 2$, where it is not necessary to go into types in full generality, this case is discussed in greater details in the following section.

\subsection{The $k$-copy most discriminable classical bits}
For concreteness and for developing a sharper intuition on the problem, we present a proof of \cref{thm:classical_dit} for the case $d = 2$ in details.
In this case, classical states can be written as $\rho_i^\textup{cl} = a_i \ketbra{0} + (1-a_i) \ketbra{1}$, where $0 \leq a_i\leq 1$ is the probability of the state being the bit $\ketbra{0}$.
That is, each classical state is uniquely identified by a single number $a_i$.
 
\begin{corollary}\label{cor:classical_bit}
The maximal $k$-copy discriminability of $N$ qubit classical states respects the upper bound 
\begin{align}
    \Omega^\textup{cl}(d=2,N,k) &\leq \frac{1}{N} \sum_{l=0}^{k} \binom{k}{l}\left( \frac{k-l}{k} \right)^{k-l}\left(\frac{l}{k} \right)^l \label{eq:upper_bound_bits_original} \\
    &= \frac{1}{N} \frac{1}{k^k}\sum_{l=0}^{k} \binom{k}{l} l^l(k-l)^{k-l} \label{eq:upper_bound_bits_intermediate}\\
    &= \frac{1}{N} \frac{k!}{k^k}  \sum_{l=0}^k \frac{k^l}{l!}\label{eq:upper_bound_bits_simplified}
\end{align}

Moreover, when $N\geq k+1$, this upper bound is saturated by a set $\{\rho_i^{\textup{cl}}\}_{i=1}^N$ where $\rho_i^{\textup{cl}} = a_i \ketbra{0} + (1 - a_i)\ketbra{1}$ with $a_i = (i-1)/k$ for $i \in \{1, \ldots, k+1\}$ and, for $i > k+1$, $\rho_i^{\textup{cl}}$ can be taken to be arbitrary states.
That is, for $N\geq k+1$,
\begin{align}
    \Omega^\textup{cl}(d=2,N,k) = \frac{1}{N} \frac{k!}{k^k}  \sum_{l=0}^k \frac{k^l}{l!}.
\end{align}
\end{corollary} 
\begin{proof}
According to \cref{eq:classical_states_multicopy}, $k$ copies of a classical state $\rho_i^{\textup{cl}} = a_i \ketbra{0} + (1 - a_i) \ketbra{1}$ can be written as
\begin{align}
    \rho_i^{\otimes k} = \sum_{j_1, \ldots, j_k \in \{0,1\}} a_i^{k-w\left(\vec{j}\right)} (1-a_i)^{w\left(\vec{j}\right)} \ketbra{j_1 \ldots j_k},
\end{align}
where $w\left(\vec{j}\right) = \sum_{l=1}^k j_l$ denotes the \textit{Hamming weight} of the bitstring $\vec{j} = (j_1, \ldots j_k)$, i.e., the number of 1's in the bitstring $\vec{j}$.
Then each state can be rewritten in an explicit permutation-invariant form
\begin{align}
    \rho_i^{\otimes k} = \sum_{l=0}^{k} a_i^{k-l} (1-a_i)^{l} \left(\sum_{\substack{j_1, \ldots, j_k \in \{0,1\} \\ w\left(\vec{j}\right)=l}} \ketbra{j_1 \ldots j_k} \right).
\end{align}

The optimal classical measurements can then be assumed to have the following form 
\begin{align}\label{eq:classical_meas}
    M_i = \sum_{l=0}^k p(i|l) \left( \sum_{\substack{j_1, \ldots, j_k \in \{0,1\} \\ w\left(\vec{j}\right)=l}} \ketbra{j_1 \ldots j_k} \right)
\end{align}
Thus, the maximal $k$-copy discriminability of a classical set can be written as (see \cref{eq:def_omega_classical})
\begin{align}
    &\Omega^\textup{cl}(d = 2, N, k) = \nonumber \\ 
    &\; \max_{\{a_i\}_{i=1}^{N}} \max_{\{p(i|l)\}_{il}} \frac{1}{N} \sum_{i=1}^N \sum_{l=0}^{k} \binom{k}{l} a_i^{k-l} (1-a_i)^{l} p(i|l). 
\end{align}
By the same arguments as in the proof of \cref{thm:classical_state_discrimination}, it follows that
\begin{align}\label{eq:max_hamming_weight}
    &\Omega^\textup{cl}(d=2, N, k) = \nonumber \\ 
    &\quad \frac{1}{N}\max_{\{a_i\}_{i=1}^{N}} \sum_{l=0}^{k} \binom{k}{l} \max_i \left[ a_i^{k-l} (1-a_i)^{l}\right],
\end{align}
where the binomial $\binom{k}{l}$ corresponds to the number of bitstrings of length $k$ with Hamming weight $l$.
In fact, the equation above could have been directly obtained from \cref{eq:omega-classical-simplified} and by noticing that one can group the optimisation over bitstrings with the same Hamming weight.

Now, the function $f_l(x) =  x^{k-l} (1-x)^{l}$ is maximised in the interval $x \in [0,1]$ for $x = \frac{k-l}{k}$ as proven in~\cite[Appendix C]{2024VieiraMilzVitagliano}.
In this way we arrive at the upper bound given by \cref{eq:upper_bound_bits_original}.
    
Moreover, this argument also provides a strategy to attain this upper bound when $N \geq k+1$, which is described in the statement of the theorem.
For each $i \in \{1, \ldots, k+1\}$, the idea is to choose $\rho_i^\textup{cl} = a_i \ketbra{0} + (1-a_i) \ketbra{1}$ with an $a_i$ that maximises the function $f_{l}$ defined above, for $l = k - i + 1$.

Finally, we will show how to simplify \cref{eq:upper_bound_bits_original}. 
By noticing that
\begin{equation}
    \left( \frac{k-l}{k} \right)^{k-l}
    \left( \frac{l}{k} \right)^l
    = \frac{(k-l)^{k-l} \, l^l}{k^k},
\end{equation}
we obtain \cref{eq:upper_bound_bits_intermediate}.
In this form, we can use the results proven in Ref.~\cite{Prodinger2013identity} to conclude that $s(k):=\sum_{l=0}^{k} \binom{k}{l} \, l^l (k-l)^{k-l} = k!\sum_{l=0}^k \frac{k^l}{l!}$, from which \cref{eq:upper_bound_bits_simplified} follows.
Interestingly, the function $s(k)$ is referred to as the \textit{Schenker sum with k-th term}, defined as A063170 at the On-Line Encyclopedia of Integer Sequences~\cite{oeisA063170}, and in Ref.~\cite{Prodinger2013identity} the authors also prove that $s(k) = k^k (1 + Q(k))$, where $Q$ is the Ramanujan $Q$-function~\cite{Flajolet1995ramanujan}.
\end{proof}

Notice that, when $N < k+1$, at least one of the measurement elements \cref{eq:classical_meas} is associated with two different Hamming weights, which leads to having not uniformly distributed on $Z$-axis states as the optimal set, for numerical details see \cref{table:classical_and_rebit}. 
This shows that, contrarily to what an initial naive intuition could suggest, the most discriminable classical bits are not bits uniformly spread between the $Z$ axis. In particular when $N<k-1$, the distribution over most discriminable bits are more concentrated towards the bits zero and one instead of being uniform.

\subsection{Asymptotic analysis  of $\Omega^\textup{cl}(d, N, k)$ and a quadratic quantum advantage}

As proven in \cref{thm:classical_dit}, for $N\geq \binom{d + k - 1}{k} $, we have that $\Omega^\textup{cl}(d,N,k) = \frac{1}{N} \textup{cap}_d(k)$. Ref.~\cite{Smith2017bounds} shows that, when $d=2$, it holds that. 
\begin{align}
 \text{cap}_2(k) &< \sqrt{\frac{\pi k}{2}} + \frac{2}{3} + \frac{1}{12}\sqrt{\frac{\pi}{2k}}, \\
 \text{cap}_2(k) &> \sqrt{\frac{\pi k}{2}} + \frac{2}{3} + \frac{1}{12}\sqrt{\frac{\pi}{2k}} - \frac{4}{135k} .
\end{align}
Hence, we see that, for large values of $k$, we have $\text{cap}_2(k) = \sqrt{\frac{\pi k}{2}} + O(1)$, consequently, for large $k$ and $N\geq k+1$, we have
\begin{align}
    \Omega^\textup{cl}(d=2,N,k) &= \frac{1}{N}\sqrt{\frac{\pi k}{2}} + O\left(\frac{1}{N}\right) \\
    &\sim \frac{1}{N} \sqrt{\frac{\pi k}{2}}. \label{eq:classical_asymptotic_bound}
\end{align}

The result can be generalised for an arbitrary dimension $d$, which we formalise in the following \cref{thm:quantum_advantage}.

We will now analyse the asymptotic behaviour of $\Omega^\textup{cl}(d,N,k) $ to show that quantum systems offer a quadratic advantage over classical states. We notice that, in the one-copy regime there is no quantum advantage over classical systems, see \cref{thm:1-copy_max_discriminable} and Refs.~\cite{Elron2007OptimalEncoding, Heinosaari_2024,2015Weiner_clinfo_in_nlevel_qusystem,2024HeinosaariHillery}. 

\begin{theorem}\label{thm:quantum_advantage}
    For a fixed dimension $d$, a large enough number of copies $k$,  and a cardinality $N$ sufficiently large to contain a state $k$-design, pure quantum states are quadratically more discriminable in $k$ than classical states:
    \begin{align}
        \Omega^\textup{cl}(d,N,k) &= \frac{1}{N}\Theta \left( \sqrt{k^{d-1}} \right), \\
        \Omega^{\textup{q}}_\textup{pure}(d,N,k) &= \frac{1}{N}\Theta \left( k^{d-1} \right) \label{eq:q_pure_asymp}.
    \end{align}
\end{theorem}

\begin{proof}
    For the pure quantum case, \cref{eq:q_pure_asymp} follows from the bounds on the combinatorial coefficient as in \cref{eq:pure_asymptotic_theta}.
    
    For the classical case, we update the derivations in the section 4D of \cite{Smith2017bounds} to be stricter; in our notation
    \begin{align}
        \text{cap}_1(k) &= \frac{k!}{k^k}\frac{k^k}{k!} = 1 = \Theta(1) = \Theta(k^{\frac{1-1}{2}}),\\
        \text{cap}_2(k) &= N\cdot \textup{\cref{eq:classical_asymptotic_bound}} = \Theta(\sqrt{k}) = \Theta(k^{\frac{2-1}{2}}),\\
        \text{cap}_d(k) &= \text{cap}_{d-1}(k) + \frac{k}{d-2} \text{cap}_{d-2}(k) \\
        &= \Theta(k^{\frac{d-2}{2}}) + \frac{1}{d-2} \Theta(k^{\frac{d-1}{2}}) = \Theta(k^{\frac{d-1}{2}}),
    \end{align}
    where the tight bound is proved by induction.
\end{proof}

\section{The $k$-copy most discriminable real states} \label{sec:real}

In this section, we study another restriction of the initial problem, where the quantum states are represented by real-valued density matrices $\{\rho_i\}_{i=1}^N\in \mathcal{L}(\mathbb{R}^d)$.
This special case is motivated by the observed quantum advantage over the classical states, which then raises the question of which conditions are necessary for the advantage to appear.
In the asymptotic setting, we manage to show that the upper bound on maximal discriminability for real qubits offers just a constant advantage over the corresponding classical value, showing that complex numbers are important for quantum advantage.

\begin{theorem} \label{thm:SO(d)}
The maximal $k$-copy discriminability of $N$ qudit real pure states respects the upper bound 
\begin{align} \label{eq:SO(d)}
 \Omega^\textup{real}_\textup{pure}(d,N,k) \leq \frac{1}{N} \left(\sum_{j=0}^{\floor{\frac{k}{2}}} m_j \sqrt{\lambda_j}\right)^2
\end{align}
where\footnote{Here, $n!!:=\prod_{l=0}^{\ceil{\frac{n}{2}}-1}(n-2l)$ stands for the double factorial, and to evaluate $m_j$ when $d=2$, we use the convention that $\binom{-1}{1} := 0$.  } 
\begin{align}
    &\lambda_j = \frac{k!(d - 2)!!}{(2j)!!(d + 2k - 2j - 2)!!} \label{eq:real_design_ev}\\
    &m_j = \binom{d + k - 2j - 1}{d - 1} - \binom{d + k - 2j - 3}{d - 1}. \label{eq:real_design_multipl_ev}
\end{align}

Moreover, when the number of states $N$ is big enough such that there exists a real qudit group $k$-design, this upper bound is saturated by a real qudit group $k$-design. That is for large enough $N$ we have that 
\begin{align} \label{eq:SO(d)_2}
 \Omega^\textup{real}_\textup{pure}(d,N,k) = \frac{1}{N} \left(\sum_{j=0}^{\floor{\frac{k}{2}}} m_j \sqrt{\lambda_j}\right)^2 .
\end{align}
\end{theorem}
\begin{proof}

Let us begin by showing the attainability of the bound by a set of states that forms a real qudit group $k$-design. As shown in Ref.~\cite{Zhou2025distinguishability}, if a set of states $\{\ketbra{\psi_i}\}_{i=1}^N$ is group covariant, it holds that 
\begin{align}
        \Discr(\{\ketbra{\psi_i}\}_{i=1}^N, k)  = \frac{1}{N} \Tr(\sqrt{\overline{\rho}})^2.
\end{align}
where $\overline{\rho}:=\sum_{i=1}^N \frac{1}{N}\ketbra{\psi_i}^{\otimes k}$.
Also, notice that, if $\lambda_j$ and $m_j$ are respectively the eigenvalues and multiplicities of $\overline{\rho}$, it holds that
\begin{align} \label{eq:mean_state}
 \Tr(\sqrt{\overline{\rho}})^2=\left(\sum_j m_j \sqrt{\lambda_j}\right)^2.
\end{align}
We now note that, if $\{\ketbra{\psi_i}\}_{i=1}^N$ is a real state group $k$-design, $\overline{\rho}=\int_{\ket{\psi}\in\mathbb{R}^d} \ketbra{\psi}^{\otimes k}\mathrm{d}\ket{\psi}$. We show the attainability by invoking Theorem 1 of Ref.~\cite{Nemoz2025real}, which proves that the eigenvalues and the multiplicities of $\int_{\ket{\psi}\in\mathbb{R}^d} \ketbra{\psi}^{\otimes k}\mathrm{d}\ket{\psi}$ are given by \cref{eq:real_design_ev} and \cref{eq:real_design_multipl_ev}, where $j\in\{0,\ldots,\floor{\frac{k}{2}}\}$.

Now, let us prove optimality. We start by noticing that any set of pure real states $\{\ketbra{\psi_i}\}_{i=1}^N$ can be written as $\{O_i\ketbra{0}O_i^{\dagger}\}_{i=1}^N$, where $O_i\in SO$ is a unitary matrix with real coefficients. We can now see that
\begin{align}
    &\max_{\{M_i\}} \sum_{i=1}^N \Tr \left( M_i \left[ O_i \ketbra{0}{0} O_i^{\dagger} \right]^{\otimes k} \right) \le  \\
    &\quad \max_{\{M_O\}} \int_{O\in SO(d)} \Tr \left( M_O \left[ O \ketbra{0}{0} O^{\dagger} \right]^{\otimes k} \right) dO, \label{eq:continuous_prob_discr}
\end{align}
where $M_O \ge 0, \int_{O\in SO(d)} M_O dO = \id$ is a valid continuous POVM.
Indeed, if the optimal POVM for the left-hand side is $\{M_i\}_{i=1}^N$, then we can always construct a POVM $\{ M_{O}\}_{i=1}^N$ for the right-hand side in a following way. If $O=O_i$ then $M_{O} = M_i \delta_i d_O$, and $M_{O} = 0$ otherwise. The constructed POVM is valid and the inequality holds.

The set of states $\{\rho_O\}_{O \in SO(d)}$, where $\rho_O = O^{\otimes k}\ketbra{0}^{\otimes k}O^{\dagger \otimes k}, \quad O \in SO(d)$, is group-covariant. From ~\cite{Chiribella_Maximum_likelihood} we know that the optimal measurement can be chosen to be group-covariant, and by \cite[Remark 2] {Chiribella_Maximum_likelihood} the optimal POVM coincides with the “square-root measurement” or the pretty-good measurements ~\cite{Hausladen01121994}. Then the optimal POVM is given by $M_O = F^{-\frac{1}{2}}\rho_OF^{-\frac{1}{2}} + (\id_{d^k}-\Pi_{\textup{sym}}(d,k))$, where $F$ is a frame operator equal to $F = \int_{O \in SO(d)} O^{\otimes k} \ketbra{0}^{\otimes k} O^{\dagger \otimes k} dO$, and its inverse is defined on its support. One can notice that the frame operator commutes with $\forall O^{\otimes k}, O \in SO(d)$ and is stable under it, i.e., $O^{\otimes k} F O^{\dagger \otimes k} = O^{\dagger \otimes k} F O^{\otimes k} = F$. By \cite[Reference 56] {2022BavarescoMuraoQuintino}, since $F$ commutes with $\forall O^{\otimes k}, O \in SO(d)$, $F^{-\frac{1}{2}}$ also commutes with $\forall O^{\otimes k}, O \in SO(d)$ as a function of $F$. For a group-covariant continuous set of states, using derivations analogous to ~\cite{Zhou2025distinguishability}, ~\cref{eq:continuous_prob_discr} can be transformed into
\begin{align}
    &\max_{\{M_O\}} \int_{O\in SO(d)} \Tr \left( M_O \left[ O \ketbra{0}{0} O^{\dagger} \right]^{\otimes k} \right) dO \\
    &\quad = \int_{O \in SO(d)} \Tr \left( F^{-\frac{1}{2}} \rho_O F^{-\frac{1}{2}} \rho_O \right) dO \\
    &\quad = \left(\Tr(F^{\frac{1}{2}})\right)^2. \label{eq:frame}
\end{align}

We can now recognise that $F=\overline{\rho}$, hence we may just combine \cref{eq:mean_state} with the spectral decomposition of $\overline{\rho}$ presented in \cite{Nemoz2025real} to finish the proof.
\end{proof}

\subsection{Pure real qubits}
Let us begin by considering only pure real states and analysing $\Omega^\textup{real}_\textup{pure}(d, N, k)$. 
For $d=2$, the sets of pure real qubits, or pure rebits, are constructed as $\{\rho_i\}_{i=1}^N=\{\ketbra{\psi_i}\}_{i=1}^N$, where $\ket{\psi_i}=\cos(\theta_i)\ket{0} + \sin(\theta_i)\ket{1}$. 

\begin{restatable}{theorem}{thmRebit}\label{thm:SO(2)}
    For real quantum states with $d=2$ and $N\geq k+1$,
        \begin{align}\label{eq:SO(2)}
        \Omega^\textup{real}_\textup{pure}(d=2,N,k) = \frac{1}{N} \frac{1}{2^k} \left( \sum_{j=0}^k \sqrt{\binom{k}{j}} \right)^2 
    \end{align}
\end{restatable}
The proof of \cref{thm:SO(2)} is presented in \cref{app:SO(2)}. The proof is based on the fact that, for $N\geq k+1$, the regular $N$-gon is a $k$-circle design, hence a real qubit design. So, for $N\geq k+1$, we can make use of \cref{thm:SO(d)}, and make use of standard combinatoric arguments to simplify \cref{eq:SO(d)} for the case where $d=2$.

For scenario $(d, N, k) = (2, 3 ,2)$ in the \cref{thm:trine_new}, we proved the optimal set of states constitutes an equiangular triangle in the XZ-plane of the Bloch sphere. The maximal $k$-copy discriminability in this case equals to $\Omega^\textup{real}_\textup{pure}(d=2,N = 3,k=2) = \frac{1}{2}+\frac{\sqrt{2}}{3} \approx 0.9714$. Furthermore, as stated in \cref{thm:trine_new}, when restricting to real states, one has for any $N\geq 3$ that
\begin{equation}
   \Omega^\textup{real}_\textup{pure}(d=2,N \geq 3,k=2)= \frac{1}{N}\left(\frac{3}{2}+\sqrt{2}\right).
\end{equation}
The proof of \cref{thm:trine_new} is presented in \cref{app:trine}; the proof is based on an alternative argument than in \cref{thm:SO(d)} and makes use of the dual problem for state discrimination. Motivated by the fact that the commutant of $SO(2)$ is spanned by the projector onto the symmetric space, anti-symmetric space, and the maximally entangled state, we present an educated guess for a feasible point for the dual problem, ensuring an upper bound. We then show that this upper bound is tight using the lower bound that comes from \cref{thm:SO(2)} tailored for $N=3$. We also conjecture that the trine is also most discriminable among pure (complex) qubits for any $k$. For $k < 5$, this conjecture is strongly supported by numerical evidences.

\subsection{Pure real qubits and regular polygons}
In \cref{sec:pure}, we recognised that when a state $k$-design exists, such states form a set that is $k$-copy most discriminable. For real states analogous results hold, i.e., when a real qudit group $k$-design exists, the maximal $k$-copy discriminability is achieved by a real qudit group $k$-design. For $d=2$ real state $k$-designs are precisely the circle $k$-designs, and it is known that a regular $N$-gon is always an $N-1$ design. Therefore, the optimal set of states consists of regular polygons in the $XZ$-plane.

We observe that, as presented in \cref{table:classical_and_rebit}, this result seems to be even stronger. Our numerical analysis suggest that, for real qubits, for any number of states $N$ and any number of copies $k$, regular $N$-gons are most discriminable rebits, even when $N<k+1$. Additionally, we have numerically verified that the optimal measurements are given by the pretty good measurements~\cite{Hausladen01121994}.

\subsection{Mixed real qubits}

Similarly to the case of complex-valued quantum states, we verify that the most discriminable real qudits are not necessarily pure. Our numerical results, see \cref{app:numerical}, shows that for $d=2$ and some values of $N$ and $k \geq 2$, we have that $\Omega^\textup{real}(d,N,k)>\Omega^\textup{real}_\textup{pure}(d,N,k)$. Also, we observe that, whenever a real state $k$-design exists for $N-1$ states, the set of $N$ states consisting of that design and a fully mixed state outperforms pure real states.

\subsection{Asymptotic analysis  of $\Omega^\textup{real}_\textup{pure}(d=2, N, k)$}

We now analyse the asymptotic behaviour of $\Omega^\textup{real}_\textup{pure}(d=2, N, k)$. In \cref{thm:SO(2)}, we prove that, for $N \geq k+1$, $\Omega^\textup{real}_\textup{pure}(d=2, N, k)= \frac{1}{N} \frac{1}{2^k} \left( \sum_{j=0}^k \sqrt{\binom{k}{j}} \right)^2$, where the bound is achieved by regular $N$-gons. 

In order to understand the asymptotic behaviour of this expression, we can study Rényi entropy for binomial distributions. In particular, using Eq.~1 of Ref.~\cite{Jacquet1999entropy} and defining the order of entropy $\omega:=1/2$, the success rate $p:=1/2$, and the number of trials $n:=k$, the Rényi entropy for binomial distributions reads as:
\begin{align}
    h_k(1/2) = \ln \left( \frac{1}{2^k} \left(\sum_{j=0}^k \sqrt{\binom{k}{j} }\right)^2 \right).
\end{align}
That is, $\frac{1}{2^k}\left( \sum_{j=0}^k \sqrt{\binom{k}{j}} \right)^2=e^{h_k(1/2)}$. 

Since the standard deviation of a single trial of the binomial distribution is given by $\sigma=1/2$, Theorem 1 of Ref.~\cite{Jacquet1999entropy} ensures that
\begin{align}
    h_k(1/2)=\ln \left( \sqrt{2\pi k} \right) + o(1),
\end{align}
hence, we have that 
\begin{align}\label{eq:lowerbound_real_entropy}
    \frac{1}{2^k}\left( \sum_{j=0}^k \sqrt{\binom{k}{j}} \right)^2 = \sqrt{2\pi k} \cdot \Theta(1) = \Theta(\sqrt{2\pi k})
\end{align}
and, for large $k$, the lower bound for rebits behaves as
\begin{align}
    \frac{1}{N} \frac{1}{2^k} \left( \sum_{j=0}^k \sqrt{\binom{k}{j}} \right)^2 = \frac{1}{N} \sqrt{2\pi k}\cdot \Theta(1).
\end{align}

This approach gives us an interesting connection of the real quantum case to Rényi entropy, but it does not allow us to obtain the exact constant in the limit of large number of copies. 
In order to do so, we note that in Ref.~\cite{polya_analysis} (problem 40 of part II), it is shown that, 
\begin{align}
    \sum_{j=0}^k \sqrt{\binom{k}{j}} \sim \frac{2^{k/2}}{\sqrt{1/2}} \left( \frac{2}{\pi k} \right) ^{-1/4} = 2^{k/2} \cdot (2\pi k)^{1/4},
\end{align}
which gives us the asymptotic equivalence for $\Omega^\textup{real}_\textup{pure}(d=2, N, k)$
\begin{align}
    \frac{1}{N} \frac{1}{2^k} \left( \sum_{j=0}^k \sqrt{\binom{k}{j}} \right)^2 \sim \frac{1}{N} \sqrt{2\pi k}.
\end{align}

With these evaluations, we can notice that, for large values of $k$, and $N\geq k+1$, classical bits respect $\Omega^\textup{cl}(d=2,N,k)\sim \frac{1}{N} \sqrt{\frac{\pi k}{2}} $, and pure rebits respect $\Omega^\textup{real}_\textup{pure}(d=2,N,k)   \sim \frac{1}{N} \sqrt{2\pi k} $. That is, rebits only offer a constant advantage over classical bits. 

This computation also shows that complex pure qubits offer a quadratic advantage over pure real qubits. 

\section{Numerical methods for upper and lower bounds on the $k$-copy discriminability} \label{sec:numerics}

\subsection{Lower bounds} \label{sec:numerics_lower}

In this subsection, we present methods used to obtain numerical lower bounds on the maximal $k$-copy discriminability of $N$ qudit states.
It is known that, for a given set of states $\{\rho_i\}_{i=1}^{N}$, its discriminability can be determined by finding optimal measurements via a semidefinite programming (SDP) approach~\cite{Boyd2004}. The formulation of the problem, as described in Ref. ~\cite{Watrous2018Book}, is 
\begin{align} \label{sdp:optimal_discr}
    \text{Given} \quad & \{\rho_i\}_{i=1}^{N} \nonumber \\
    \text{maximize} \quad & \frac{1}{N} \sum_{i=1}^{N} \Tr(\rho_i^{\otimes k} M_i) \nonumber \\
    \text{subject to} \quad & M_i \geq 0 \quad \forall i \nonumber \\
    & \sum_{i=1}^{N} M_i = \id
\end{align}

This is a convex optimisation of a linear function over affine and positive semidefinite constraints. It involves finding $N$ operators of size $d^k \times d^k$. For this semidefinite program, by Slater’s condition, strong duality holds, and the optimum value is attained. Nevertheless, finding the maximal $k$-copy discriminability of $N$ qudit states cannot be solved with a one-shot SDP, since the maximisation over qudit states is not convex. We used several techniques for finding the lower bound of an $\Omega(d, N, k)$.

\textit{Grid search with SDP.} For qubit case $d=2$, we can discretize the set of possible states using the Bloch sphere representation and parameterize each state by two angles: a polar angle $\theta_i \in [0, \pi]$ and an azimuthal angle $\phi_i \in [0, 2\pi)$. Then the optimisation \cref{sdp:optimal_discr} is performed for every possible set of $N$ states. This exhaustive grid search scales with the cardinality of a search space of such sets, given by $\binom{m}{N}$, where $m$ denotes the number of candidate states determined by the discretisation step size of the polar and azimuthal angles.

The run-time complexity of \cref{sdp:optimal_discr} scales exponentially with the number of state copies $k$ available, and with the number of discretisation steps. For pure states, there is a possible improvement suggested in Ref. ~\cite{mohan2023, Boyd2004}, which reduces the size of the standard SDP for the minimum-error state discrimination problem. The approach consists of constructing a Gram matrix for a given set of pure states $\{\ketbra{\psi_i}\}_{i=1}^{N}$. The Gram matrix is an $N \times N$ matrix with entries $G_{ij} = \bra{\psi_i}\ket{\psi_j}^k$. The initial problem can thus be reformulated as
\begin{align}
    \text{Given} \quad & \{\ketbra{\psi_i}\}_{i=1}^{N} \nonumber \\
    \text{maximize} \quad & \frac{1}{N} \sum_{i=1}^{N} \bra{i}W_i\ket{i} \nonumber \\
    \text{subject to} \quad & W_i \geq 0 \quad \forall i \nonumber \\
    & \sum_{i=1}^{N} W_i = G.
\end{align}

This optimisation finds $N$ operators of size $N \times N$, and is independent of the number of available copies $k$, providing a substantial decrease in the runtime complexity compared to the previous approach.

Grid search provides a lower bound on the value of a function $\Omega(d, N, k)$, since it relies on the discretisation of the search space. Nevertheless, with this approach, it is easy to control and constrain the search space if needed. Combined with the Gram matrix formulation it is a useful tool for the systematic exploration of possible candidate solutions, which provides a good understanding of the structure of the problem.

\textit{Gradient descent.} Another approach to obtaining a numerical lower bound of $\Omega(d, N, k)$ is to use gradient descent methods, in particular, the Adam algorithm \cite{kingma2017adam}. This is a stochastic optimisation suitable for non-convex optimisation problems. In our case, the parameters in the algorithm are defined as $2N$ matrices of size $d \times d$ for real and imaginary values of every density matrix of a state in the set, and $2N$ matrices of size $d^k \times d^k$ for measurements. The parameters are optimised using a finite differences approximation of the first derivative with standard hyperparameters: exponential decay rates $\beta_1=0.9$ and $\beta_2=0.999$, exponentially decreasing learning rate from $\alpha=0.1$ to $10^{-8}$, $10^4$ iterations, and several trials with random initial values. The algorithm performs a search over the whole unconstrained space and provides high precision of the value of the objective function; more details about how we unconstrained the problem are provided in \cref{app:gd_details}.

All of the obtained numerical results are presented in \cref{table:numerics}, \cref{table:classical_and_rebit} in the appendix. The code is written in the Julia language and uses the MOSEK solver~\cite{mosek}, and is available on GitHub~\cite{kvashchuk_github}.

\subsection{Upper bounds} \label{sec:numerics_upper}

We now introduce a hierarchy of semidefinite programming relaxations providing upper bounds on the maximum discriminability $\Omega$ for several scenarios. This section outlines the SDP formulation, with more details on the implementation found in \cref{app:sdp_details}.

Let $\mathcal{H}_S$ and $\mathcal{H}_M$ be the Hilbert spaces associated with the states $\rho_i$ and POVM elements $M_i$, respectively. In order to optimize over both states and measurements simultaneously, we make use of the identity $\Tr(F_{BA} A \otimes B) = \Tr(A B)$, where $F_{BA}$ is the swap or ``flip'' operator acting on the Hilbert spaces of operators of $A$ and $B$. First, we apply the identity to rewrite the trace in \cref{eq:def_omega} as
\begin{equation}\label{eq:sdp_fliptrace}
    \Tr(\rho_i^{\otimes k} M_i) = \Tr(F_{MS} \rho_i^{\otimes k} \otimes M_i).
\end{equation}
This expression is neither linear nor convex in terms of $\rho_i$ and $M_i$, making it unsuitable in a semidefinite program's objective function. To address this, we first replace $\rho_i^{\otimes k} \otimes M_i$ with the operators $\{ \Phi^0_i \}_{i=1}^N$ satisfying the constraints
\begin{align}\begin{split}\label{eq:sdp_phi0}
    \Phi^0_i \ge 0,\quad \Phi^0_i \in \operatorname{Perm}_d^k \otimes \operatorname{Perm}_d^k, \\
    \sum_{i=1}^N \Tr_S(\Phi^0_i) = \id_M, \quad \Phi^0_i \in \mathrm{Sep}(\mathcal{P}_{1^k|k})
\end{split}\end{align}
where $\mathrm{Sep}(\mathcal{P}) \subset \mathcal{L}(\mathcal{H}_S^{\otimes k} \otimes \mathcal{H}_M)$ is the convex cone of separable operators according to the partition $\mathcal{P}$, where $\mathcal{P}_{1^k|k} = \{ (1), (2), \ldots, (k), (k+1, \ldots, 2k) \}$.

Note that since $\rho_i^{\otimes k} \in \operatorname{Perm}_d^k$, only the projection of $M_i$ onto $\operatorname{Perm}_d^k$ contributes to the trace. Furthermore, since no constraints on $M_i$ act nontrivially outside this space, without loss of generality we can also enforce permutation invariance on $M_i$.

To approximate the separability constraint in $\Phi^0_i$, we make use of the Doherty-Parrilo-Spedalieri~\cite{dps2004_complete} (DPS) hierarchy (or, more precisely, the quantum de Finetti theorem~\cite{caves2002_quantumdefinetti}) and the positive partial transpose (PPT) criterion~\cite{peres1996_ppt}. This is done by considering instead positive operators $\Phi^\ell_i \in \operatorname{Perm}_d^{k+\ell} \otimes \operatorname{Perm}_d^k$ acting as potential extensions of $\Phi^0_i$ with $\ell$ additional permutation-invariant copies of the state. By the quantum de Finetti theorem, we then have
\begin{align}
	\min_{\Phi^0_i} \| \Tr_{S_{\ell}}(\Phi^\ell_i) - \Phi^0_i \|_1 \le \epsilon(d, k, \ell)
\end{align}
where $\Tr_{S_{\ell}}$ denotes the partial trace over $\ell$ copies of the state and $\epsilon(d, k, \ell)$ a nonnegative function such that $\lim_{\ell \to \infty} \epsilon(d, k, \ell) = 0$. Purity constraints and proper normalisation on the state subspace are obtained by introducing a new optimisation variable $\Phi^\ell_{\perp i}$ corresponding to $\rho_i^{\otimes k} \otimes (\id - M_i)$, and an auxiliary variable $\Psi_i \in \mathcal{L}(\mathcal{H}_S^{\otimes k+\ell})$ such that $\Phi^\ell_i + \Phi^\ell_{\perp i} = \Psi_i \otimes \id_M$. The final SDP relaxation becomes:

\begin{equation}\begin{split}\label{eq:sdp_main}
	&\text{maximize}\quad \frac{1}{N} \sum_{i=1}^N \Tr[ \left( \id_{S_\ell} \otimes F_{M S_k} \right) \Phi^\ell_i ] \\
	&\text{subject to} \\
	&\quad \Phi^\ell_i, \Phi^\ell_{i \perp}, \Psi_i \ge 0, \\
	&\quad \Phi^\ell_i + \Phi^\ell_{\perp i} = \Psi_i \otimes \id_M, \\
	&\quad \Phi^\ell_i, \Phi^\ell_{\perp i} \in \operatorname{Perm}_d^{k+\ell} \otimes \operatorname{Perm}_d^k, \\
	&\quad \sum_{i=1}^N \Tr_S(\Phi^\ell_i) = \id_M,\quad \Tr(\Psi_i) = 1, \\
	&\quad (\Phi^\ell_i)^{\mathrm{T}_n}, (\Phi^\ell_{\perp i})^{\mathrm{T}_n} \ge 0,\; n = 1, \ldots, k+\ell.
\end{split}\end{equation}
Here, $\mathrm{T}_n$ denotes the partial transpose with respect to $n$ copies of $\mathcal{H}_S$, which can be chosen arbitrarily due to permutation invariance. For the scenario involving pure states, we let $\psi_i = \Tr_{S_{k+\ell-2}}(\Psi_i)$ and further impose $\Tr(F_{S_2 S_1} \psi_i) = 1$, corresponding to the condition $\Tr(\rho_i^2) = 1$. Upper bounds for rebits can be obtained by enforcing $\mathrm{Im}[\Psi_i] = 0$. Further details on the implementation can be found in \cref{app:sdp_details}.

For this work, we have employed this SDP to obtain upper bounds on $\Omega$ for the case $d = 2$ and several scenarios involving pure and mixed qubits and rebits, with $k+\ell \le 5$. Results can be found in \cref{{table:upper_numerics_sdp}}. Optimisations were performed using Python and CVXPY~\cite{diamond2016cvxpy,agrawal2018rewriting}, with MOSEK~\cite{mosek} and SCS~\cite{odonoghue21_SCS} solvers, running on an Intel Core Ultra 9 285K CPU with 128 GB of RAM. The code is available on GitHub~\cite{1ucasvb_github}.

Note that \cref{eq:sdp_phi0} and the SDP in \cref{eq:sdp_main} relax the product form necessary for the substitution in \cref{eq:sdp_fliptrace}. The resulting feasible region for the operators $\Phi^\ell_i$ includes convex mixtures of strategies, effectively allowing for classically correlated state preparations and measurements. It is currently unknown if this additional resource leads to an increased value for the maximum discriminability.

\section{Visual illustration of the most discriminable two-dimensional states} \label{sec:visualisation}
In this section, we present a Bloch sphere representation for the most discriminable states for some small instances of the problem considered here \cref{fig:k2_N4_pure_mixed}-\cref{fig:k3_N7_pure_mixed}. We recall that the solution of this problem is not unique. For the pure case, when there exists a state $k$-design with $N-1$ states, we may always construct a set of the most discriminable states where $N-1$ states form a design, and the other state is arbitrary, as detailed in the proof of \cref{thm:k-copy-optimal-pure-design}. This explains why in \cref{fig:k2_N5_pure_mixed} we represent only $4$ states instead of $5$, and analogously in \cref{fig:k3_N7_pure_mixed} only $6$ states instead of $7$. Also, for some instances, when for given parameters the design has not appeared yet, solutions are also not unique; for example, in \cref{fig:k3_N5_pure_mixed}, for the pure case, the bipyramid attains this value, but it is not the only solution. To characterize the entire set of sets with this discriminability is a nontrivial task.

\begin{figure}[h!]
\centering
\begin{subfigure}{0.45\linewidth}
    \centering
    \includegraphics[width=\linewidth]{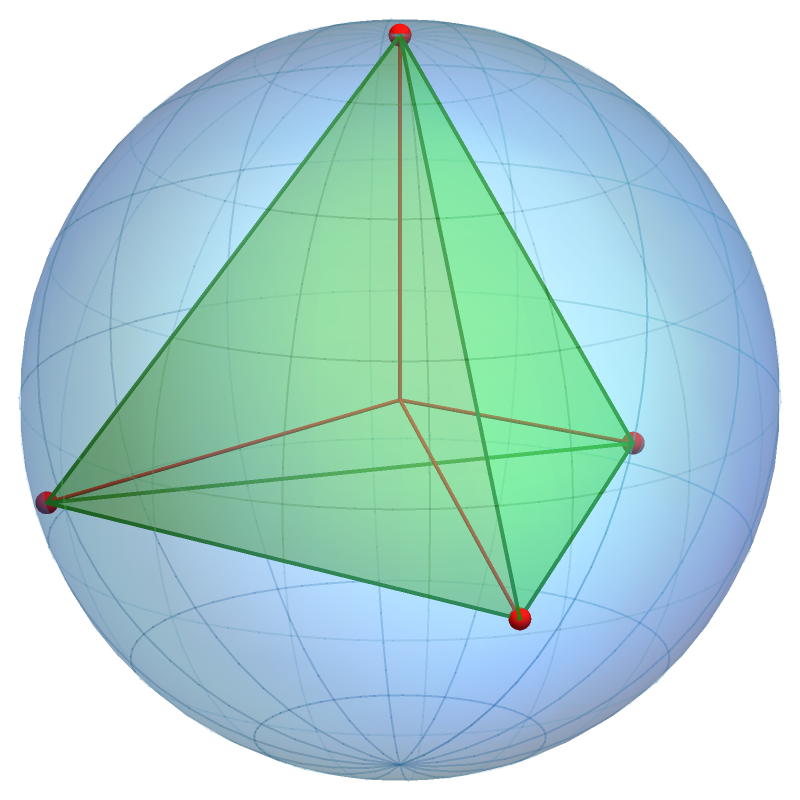}
    \caption{$\Omega^{\textup{q}}_{\textup{pure}}(2,4,2) = 0.75$}
\end{subfigure}
\hspace{0.03\linewidth}
\begin{subfigure}{0.45\linewidth}
    \centering
    \includegraphics[width=\linewidth]{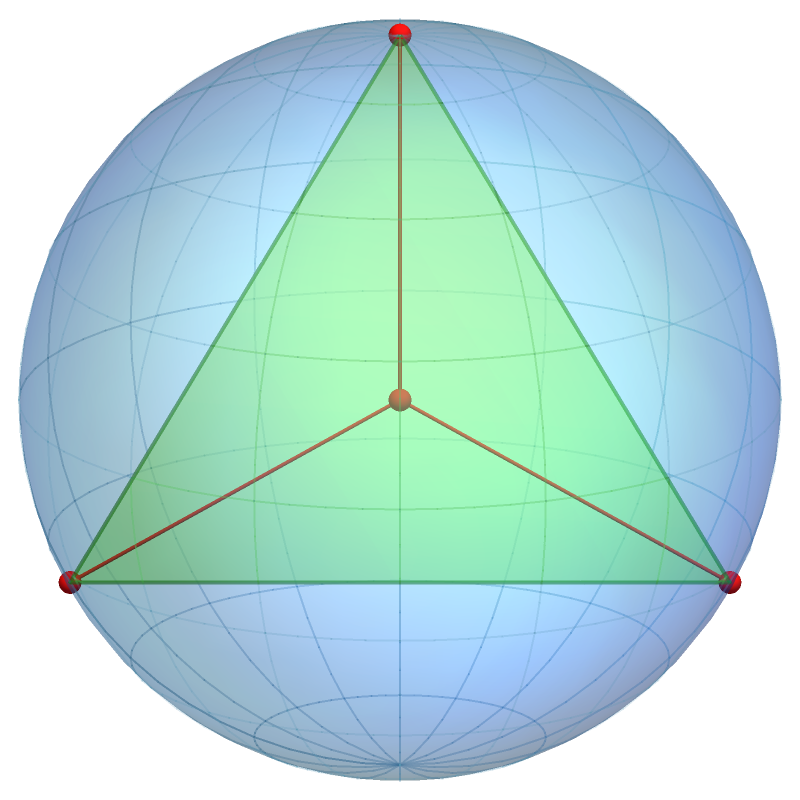}
    \caption{$\Omega^{\textup{q}}(2,4,2) = 0.7911$}
\end{subfigure}
\caption{$(d, N, k) = (2, 4, 2)$}
\label{fig:k2_N4_pure_mixed}
\end{figure}

\begin{figure}[h!]
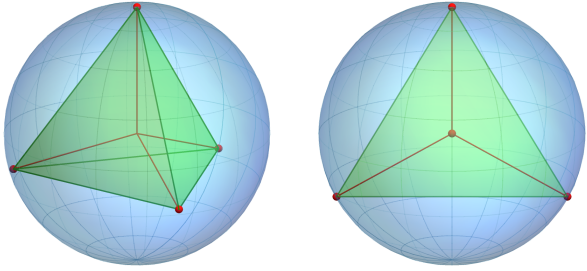

\centering
\begin{subfigure}{0.45\linewidth}
    \centering
    \includegraphics[width=\linewidth]{figures/bloch_tetrahedron.png}
    \caption{$\Omega^{\textup{q}}_{\textup{pure}}(2,4,3) = 0.9714$}
\end{subfigure}
\hspace{0.03\linewidth}
\begin{subfigure}{0.45\linewidth}
    \centering
    \includegraphics[width=\linewidth]{figures/bloch_tetrahedron.png}
    \caption{$\Omega^{\textup{q}}(2,4,3) = 0.9714$}
\end{subfigure}
\caption{$(d, N, k) = (2, 4, 3)$}
\label{fig:k3_N4_pure_mixed}
\end{figure}

\begin{figure}[h!] 
\centering
\begin{subfigure}{0.45\linewidth}
    \centering
    \includegraphics[width=\linewidth]{figures/bloch_tetrahedron.png}
    \caption{$\Omega^{\textup{q}}_{\textup{pure}}(2,5,2) = 0.6$}
\end{subfigure}
\hspace{0.03\linewidth}
\begin{subfigure}{0.45\linewidth}
    \centering
    \includegraphics[width=\linewidth]{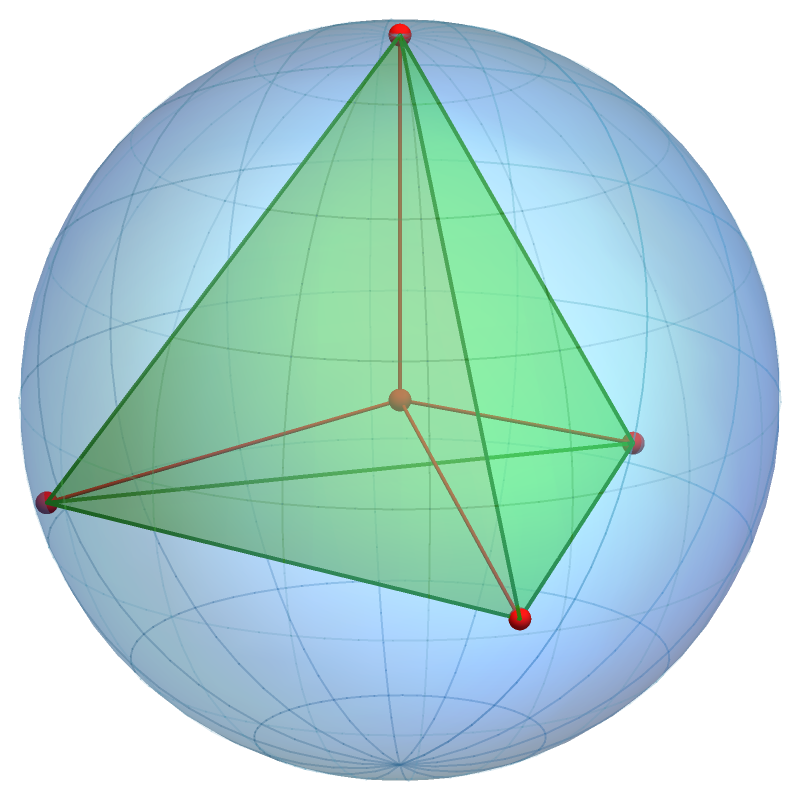}
    \caption{$\Omega^{\textup{q}}(2,5,2) = 0.65$}
\end{subfigure}
\caption{$(d, N, k) = (2, 5, 2)$}
\label{fig:k2_N5_pure_mixed}
\end{figure}

\begin{figure}[h!]
\centering
\begin{subfigure}{0.45\linewidth}
    \centering
    \includegraphics[width=\linewidth]{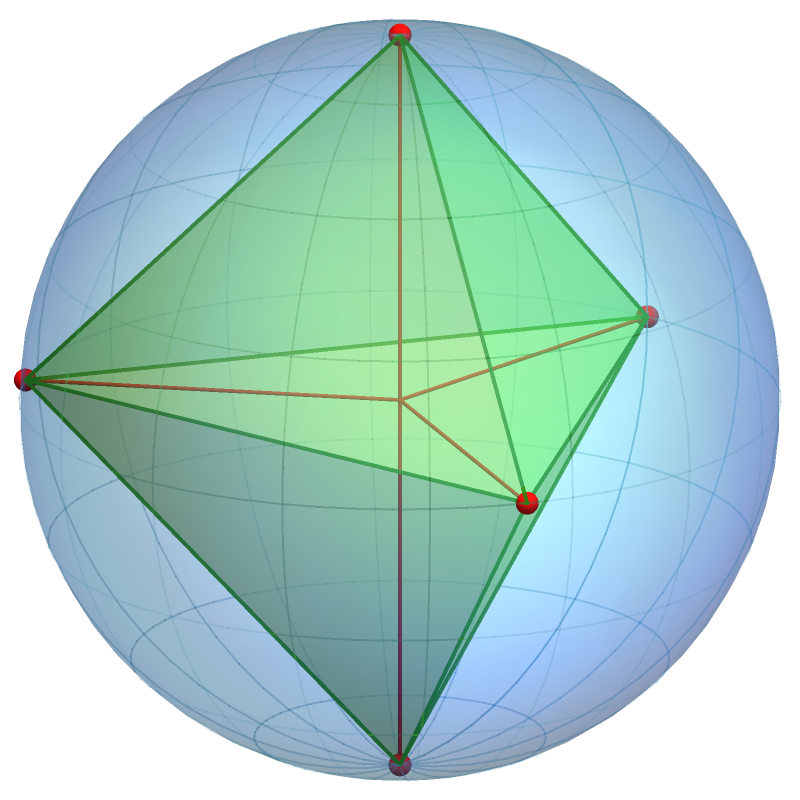}
    \caption{$\Omega^{\textup{q}}_{\textup{pure}}(2,5,3) = 0.7901$}
\end{subfigure}
\hspace{0.03\linewidth}
\begin{subfigure}{0.45\linewidth}
    \centering
    \includegraphics[width=\linewidth]{figures/bloch_tetrahedron_center.png}
    \caption{$\Omega^{\textup{q}}(2,5,3) = 0.8771$}
\end{subfigure}
\caption{$(d, N, k) = (2, 5, 3)$}
\label{fig:k3_N5_pure_mixed}
\end{figure}

\begin{figure}[h!]
\centering
\begin{subfigure}{0.45\linewidth}
    \centering
    \includegraphics[width=\linewidth]{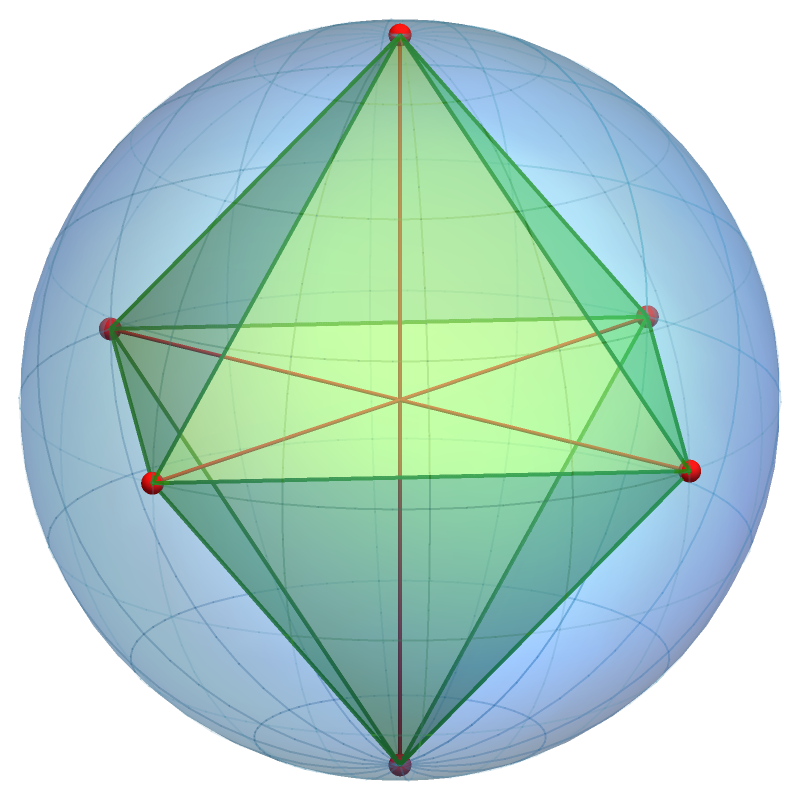}
    \caption{$\Omega^{\textup{q}}_{\textup{pure}}(2,6,3) = 0.6667$}
\end{subfigure}
\hspace{0.03\linewidth}
\begin{subfigure}{0.45\linewidth}
    \centering
    \includegraphics[width=\linewidth]{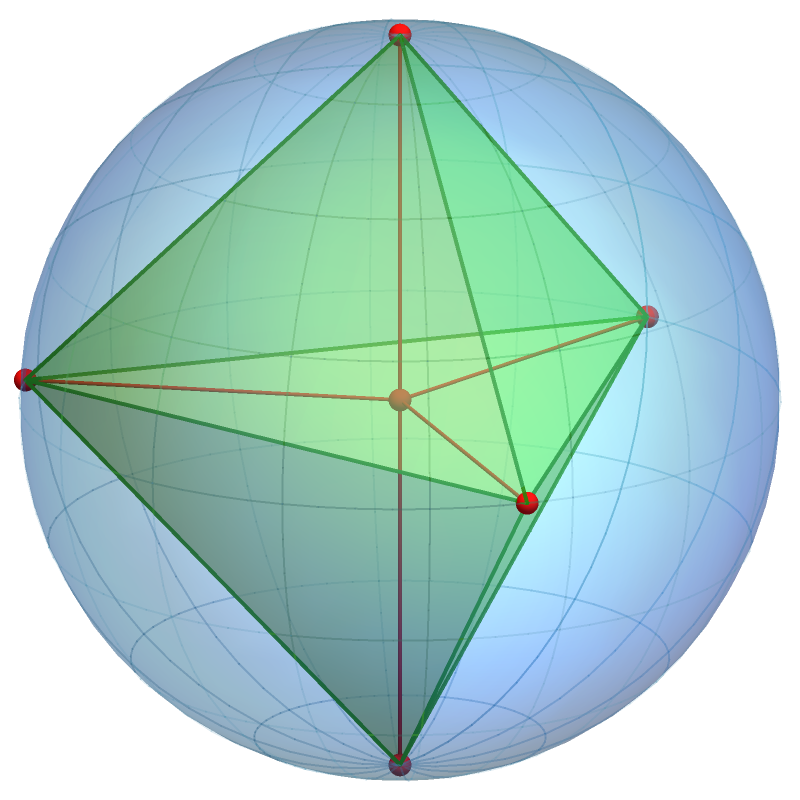}
    \caption{$\Omega^{\textup{q}}(2,6,3) = 0.7418$}
\end{subfigure}
\caption{$(d, N, k) = (2, 6, 3)$}
\label{fig:k3_N6_pure_mixed}
\end{figure}

\begin{figure}[h!]
\centering
\begin{subfigure}{0.45\linewidth}
    \centering
    \includegraphics[width=\linewidth]{figures/bloch_bipyramid4.png}
    \caption{$\Omega^{\textup{q}}_{\textup{pure}}(2,7,3) = 0.5714$}
\end{subfigure}
\hspace{0.03\linewidth}
\begin{subfigure}{0.45\linewidth}
    \centering
    \includegraphics[width=\linewidth]{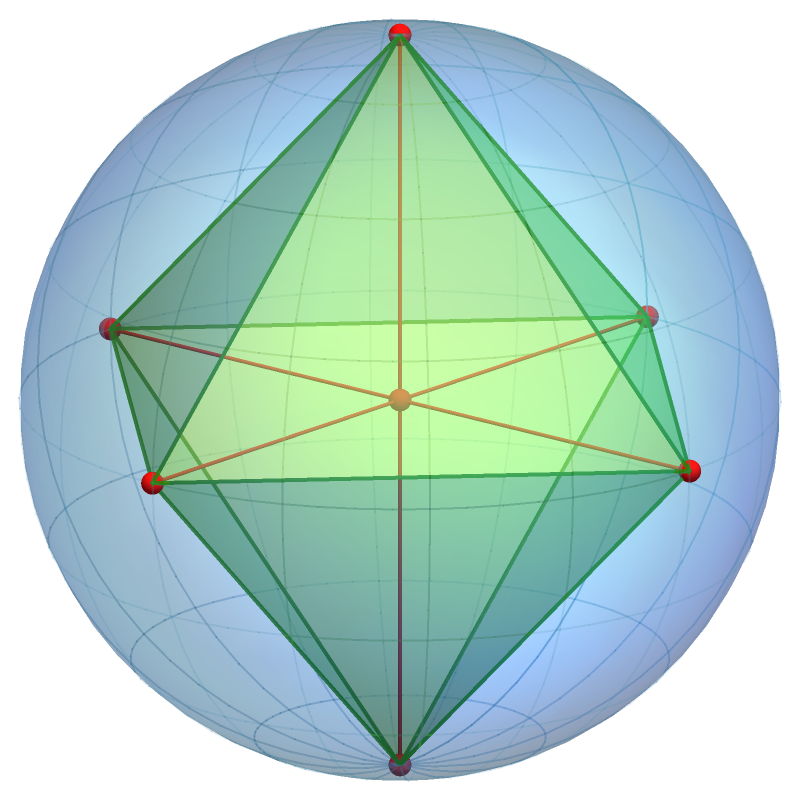}
    \caption{$\Omega^{\textup{q}}(2,7,3) = 0.6429$}
\end{subfigure}
\caption{$(d, N, k) = (2, 7, 3)$}
\label{fig:k3_N7_pure_mixed}
\end{figure}

\section{The relationship between $k$-copy discrimination problem and state $k$-designs} \label{sec:relationship}

In this section we summarise the results obtained in this work that connect the problem of finding the most discriminable $k$-copy states with the concept of state $k$-designs. 

\subsection{Pure quantum states} \label{subsec:pure_and_design}
When focusing on pure quantum states (see \cref{sec:pure}), the relationship between problems of finding the most discriminable $k$-copy states and finding state $k$-designs is very strong. More precisely, as stated in \cref{coro:omega_pure}, when the number of states $N$ is large enough to contain a state $k$-design, then maximal $k$-copy discriminability $\Omega^\textup{q}_\textup{pure}(d,N,k)$ is attainable. In this case, a set of pure states is most discriminable if and only if it forms a weighted state $k$-design, and such a set has maximal $k$-copy discriminability, which is given by $\Discr(\{\ketbra{\psi_i}\}_{i=1}^N, k) =  \frac{1}{N} \binom{d + k - 1}{k}$. This ensures that our computational methods to find lower and upper bounds on the function $\Omega^\textup{q}_\textup{pure}(d,N,k)$ are also methods to find state $k$-designs, or to rule out state $k$-designs for some given parameters $(d,N,k)$. 

Now, when the number of states $N$ is not large enough to contain a state $k$-design, the problem of finding the most discriminable $k$-copy states becomes more subtle. While \cref{thm:epsilon} shows that $\epsilon$-approximate state $k$-designs approximate the most discriminable $k$-copy states, the problem of finding the most discriminable states in this regime is likely a hard computational task. We notice that, for a given number of candidates $N$, dimension $d$, the problem of finding $k$-state designs, or even ensuring whether they exist, is a hard mathematical and computational problem~\cite{An2026survey}. Since checking whether $\Omega^\textup{q}_\textup{pure}(d,N,k) =  \frac{1}{N} \binom{d + k - 1}{k}$ is a necessary and sufficient condition for the existence of a state $k$-design, the problem of evaluating the function $\Omega^\textup{q}_\textup{pure}(d,N,k)$ and finding states that saturate this bound is also mathematically and computationally hard.

Finally, we mention that when the number of states $N$ is not large enough to contain a state $k$-design, we expect that the sets of pure states for which $\Discr(\{\ketbra{\psi_i}\}_{i=1}^N, k) = \Omega^\textup{q}_\textup{pure}(d,N,k)$ are likely good approximations of state $k$-designs, and they can be used for applications where state $k$-designs are useful, such as in quantum tomography, quantum cryptography, and quantum error correction.

\subsection{Mixed quantum states} 
As discussed in \cref{sec:mixed}, when the number of states $N$ is strictly larger than the required number of states to contain a $k$-design, the most discriminable $k$-copy states are not pure and are not strictly given by a state $k$-design. However, in this regime, we have shown that some advantageous $k$-copy states are still closely related to state $k$-designs, since they can be obtained from a set of states forming a $k$-design and the maximally mixed state. This construction follows from the intuition that, the optimal measurement for most discriminable $k$-copy quantum states, which by \cref{thm:k-copy-optimal-pure-design} forms a weighted state $k$-design, does not add up to the identity map in $\left(\mathbb{C}^d\right)^{\otimes k}$, but to the projector onto the symmetric subspace $\operatorname{Sym}_d^k$, so the optimal measurement is not normalised. This allows for an extra measurement outcome, defined as $M_{N+1} = \id - \Pi_\textup{sym}(d,k)$, which can increase the success probability of the discrimination task.

While we do not have proof that this construction is optimal, for the qubit case we have numerical evidence that it is. We expect that, at least for $d=2$, this construction may be optimal when $N = N'(d, k) + 1$, so it is large enough to contain a state $k$-design on $N'(d,k)$ and one extra state.

\subsection{Real quantum states} 
As discussed in \cref{sec:real}, analysing multiple copies of real quantum states is mathematically more complicated than its complex counterpart, as real state $k$-designs do not add up to a projector, but to an operator with several different eigenvalues. Since real state $k$-designs do not add up to a projector, the methods used in the complex case to obtain upper bounds do not apply in a straightforward way, but we could still prove that when a real group state $k$-design exists, the set of states containing it has the maximal $k$-copy discriminability, see \cref{thm:SO(d)}. For the qubit case, we have proven that, for $N\geq k+1$ and large values of $k$, we have $ \Omega^\text{real}_\text{pure}(d=2,N,k) \sim \frac{1}{N} \sqrt{2\pi k}$ while for the classical case we have $\Omega^\text{cl}(d=2,N,k) \sim \frac{1}{N}  \sqrt{\frac{\pi k}{2}}$. That is, the advantage of encoding classical information on real pure qubits over classical bits is constant, while as detailed in \cref{subsec:asympt_pure}, the advantage given by encoding classical information on complex qudits is quadratic. This is evidence that the power of real quantum states is very limited~\cite{Nemoz2025real}.

It is also interesting to note that we have strong numerical evidence that the regular polygon is always the most discriminable set of pure real qubits, regardless of the number of states $N$ or the number of copies $k$. While this seems to be intuitively true to us, we were unable to find a proof of this fact.

Lastly, we notice that, as in the complex case, when mixed states are considered, the most discriminable real quantum states are not pure. And, as in the complex case, we expect that, while state $k$-designs are not the most discriminable $k$-copy states, they will be still closely related to the optimal construction.

\section{Discussions} 

Despite the results obtained on maximal $k$-copy discriminability of different sets of states, including quantum pure and mixed states, real states, and classical states, and showing the relation between the quantum state discrimination problem and quantum state $k$-design, we can still enumerate several open questions.

\begin{enumerate}
    \item What is the asymptotic behaviour of the function $\Omega_\textup{pure}^{\textup{real}}(d,N,k)$ for $d>2$? We have shown that, for $d=2$, real qubits offer only a constant advantage over classical bits (contrary to complex qubits, where this advantage is quadratic). Do real qudits offer only a constant advantage over classical bits for any dimension $d$?
    \item Is the regular $N$-gon set of rebits always the most discriminable set of pure rebits for any number of copies $k$? Apart from the numerical evidence, we do not have proof of this fact.
    \item For $d=2$ and $N=3$, is the regular triangle the most discriminable set of qubits for any number of copies $k$? We could prove it for $k=2$, but it is not clear to us how to extend this proof for any number of copies $k$.
     \item For mixed states and $N$ large enough to exactly contain a state $k$-design and only one more state, is the construction of the most discriminable $k$-copy states from a state $k$-design combined with the maximally mixed state optimal for any dimension? We have numerical evidence that it is true for qubits, but we do not have a proof of this fact. Also, this construction seems to be useful even when a design does not exist (see e.g., \cref{subsection:mixed_examples} that addresses the scenario where $d=2$, $N=4$, and $k=2$.)
     \item The numerical upper bound method presented in this work is not known to converge for every scenario, but, we can verify that it does converge in some cases. We leave for future work the analysis of the convergence of this method, and to understand in which scenarios it converges and in which it does not, and why we have this phenomenon.
     \item Does the quantity $\Omega^{\textup{q}}(d, N, k)$ play a role relevant in the analysis of the information leakage, or any quantum cryptographic protocol? Since its classical analogue is a relevant and well-studied quantity for the analysis of the information leakage in classical cryptographic protocols, we expect that its quantum analogue may also have analogous interpretations and applications.
\end{enumerate}

Given the rich structure of the problem of finding the most discriminable $k$-copy states, we foresee some interesting future research directions. 
\begin{enumerate}
    \item How would the maximal $k$-copy discriminability function $\Omega(d, N, k)$ behave when considering generalised probabilistic theories (GPTs)~\cite{2023GPTs_plavala, 2005GPTs_barrett}, which are a general framework for describing physical theories, including classical and quantum mechanics. Such investigation was done for the single copy case~\cite{Heinosaari_2024}, but we expect to find different and interesting results for the multicopy case, given that GPTs with a great freedom on quantum states have severe restrictions on the possibility of joint measurements~\cite{2024joint_measurements_dmello_Ligthart, 2009joint_measurements_Skrzypczyk, 2005joint_measurements_Short}  (see however \cite{2026joint_measurements_dmello_gross}). Understanding the behaviour of the function $\Omega(d, N, k)$ for GPTs may allow us to understand the role of the structure of the state space in the problem of state discrimination with many copies, and to identify the specific features of quantum mechanics that lead to the observed results. 
    \item While here we have restricted our analysis to states, this question is also interesting to analyse for measurements. For instance, we can ask what is the most discriminable $k$-copy measurement. We notice that this question may be nontrivial already for the case of $k=1$, and may be connected to the problem of finding the measurements with most information storability~\cite{Heinosaari_2024, 2018information_storability_Matsumoto}, and the most incompatible measurements~\cite{2017incompatible_measurements_Bavaresco_Quintino}.
    \item It would be interesting to analyse the case where the prior distribution is not uniform, and to see how this affects the results obtained in this work.
\end{enumerate}

We hope the results and methods presented in this work will spark further research on the problem of state discrimination with many copies, and its connections to state $k$-designs, and possibly, to find applications and connections to other branches of quantum information.

\let\oldaddcontentsline\addcontentsline
\renewcommand{\addcontentsline}[3]{} 
\section*{Note added}
\let\addcontentsline\oldaddcontentsline 

During the final stages of this work, we became aware of a related manuscript by Achenbach \textit{et al}, titled ``Nonclassical traits in multi-copy state discrimination''~\cite{Achenbach2025}. The two works were carried out independently and address closely related questions.

\let\oldaddcontentsline\addcontentsline
\renewcommand{\addcontentsline}[3]{} 
\section*{Acknowledgments}
\let\addcontentsline\oldaddcontentsline 

We thank Peter Brown, Mariami Gachechiladze, Carlos Humberto Vieira, and Tristan Nemoz for fruitful discussions, in particular, we thank Peter Brown, Carlos Humberto Vieira, and Tristan Nemoz for discussions that were fundamental for proving the upper bound of \cref{thm:SO(d)}. We thank Hyunho Cha for pointing us to Ref.~\cite{cha}. We thank Camiel Meijer, Kristin Sundal Lien, and Raman Choudhary for comments on a first version of this manuscript.

This work has been supported by Region Île-de-France in the framework of DIM QuanTiP.
This work was funded by QuantEdu France, a State aid managed by the French National Research Agency for France 2030 with the reference ANR-22-CMAS-0001
MTQ is supported by the Agence Nationale de la Recherche (ANR) through the JCJC programme under grant number ANR-25-CE47-6396-01-HOQO-KS. LV is supported within the QuantERA II Programme that has received funding from the EU's H2020 research and innovation programme under the GA No 101017733. TO is supported by the Swedish Research Council (Grant No. 2024-05341).

\let\oldaddcontentsline\addcontentsline
\renewcommand{\addcontentsline}[3]{} 
\section*{AI disclosure}
\let\addcontentsline\oldaddcontentsline 

All theorems, results, computational code, figures, ideas, and discussions for this work were derived by the authors without the use of AI. The first two arXiv versions of this manuscript were completed before the authors had access to LLM. For the third version of manuscript (this version), AI tools (OpenAI's GPT-6 Astra ) were used to assist with reviewing the proofs and checking the grammar of the text. Some limited corrections to the mathematical arguments suggested by an LLM were verified by the authors and incorporated into the current version of the manuscript.
The authors take full responsibility for the content of the work.



%

\clearpage
\appendix

\onecolumn

\section{Summary of results for two-dimensional systems} \label{app:summary}

The table below presents the values of $\Omega(d=2,N,k)$, including analytical and numerical calculations.

\begin{table*}[h!]
\begin{center}
\begin{adjustbox}{max width=\linewidth}
\footnotesize
\setlength{\tabcolsep}{8pt}
\renewcommand{\arraystretch}{1.7}
\begin{tabular}{|c|Sc|Sc|Sc|Sc|Sc|} 
\hline
 $N$ & $\Omega^\textup{cl}$ & $\Omega_{\textup{pure}}^\textup{real}$ & $\Omega^{\textup{q}}_{\textup{pure}}$ & $\Omega^{\textup{q}}$  \\ 
\hline

\multicolumn{5}{|c|}{k = 2} \\
\hline

3 & $ \frac{1}{3} \cdot 2.5 $ & $   \frac{1}{3} \cdot 2.9142$ & $ \frac{1}{3} \cdot 2.9142$ & $\frac{1}{3} \cdot 2.9142$ \\
\hline

4 & $ \frac{1}{4} \cdot 2.5 $ & $   \frac{1}{4} \cdot 2.9142$ & $ \frac{1}{4} \cdot 3$ & $\geq \frac{1}{4} \cdot 3.1642 $ \\
\hline

5 & $ \frac{1}{5} \cdot 2.5 $ & $  \frac{1}{5} \cdot 2.9142$  & $\frac{1}{5} \cdot 3$  & $\geq \frac{1}{5} \cdot 3.2500 $   \\
\hline

6 & $ \frac{1}{6} \cdot 2.5 $ & $ \frac{1}{6} \cdot 2.9142$  & $\frac{1}{6} \cdot 3$ & $\geq \frac{1}{6} \cdot 3.2500 $   \\
\hline

7 & $ \frac{1}{7} \cdot 2.5 $  & $ \frac{1}{7} \cdot 2.9142$  & $\frac{1}{7} \cdot 3$ & $\geq \frac{1}{7} \cdot 3.2500 $  \\
\hline
$N>$ 4 & $ \frac{1}{N} \cdot \frac{5}{2} $  & $  \frac{1}{N}\left(\frac{3}{2}+\sqrt{2}\right)$  & $\frac{1}{N} \cdot 3$ & $\geq \frac{1}{N} \left(3+\frac{1}{4}\right) $  \\
\hline

\multicolumn{5}{|c|}{k = 3} \\
\hline

3 & $\frac{1}{3} \cdot 2.7501^*$ & $\frac{1}{3} \cdot 2.9747^*$ &  $\frac{1}{3} \cdot 2.9747^*$  & $\geq \frac{1}{3} \cdot 2.9747$    \\
\hline

4 & $\frac{1}{4} \cdot 2.8889 $ & $\frac{1}{4} \cdot 3.7320$ & $\frac{1}{4} \cdot 3.8856^*$ & $\geq \frac{1}{4} \cdot 3.8856$    \\
\hline

5 & $\frac{1}{5} \cdot 2.8889 $ & $\frac{1}{5} \cdot 3.7320$ & $\frac{1}{5} \cdot 3.9505^*$ & $\geq \frac{1}{5} \cdot 4.3855$  \\
\hline

6 & $\frac{1}{6}  \cdot 2.8889 $& $\frac{1}{6} \cdot 3.7320$ & $\frac{1}{6} \cdot 4$ & $\geq \frac{1}{6} \cdot 4.4508$  \\
\hline

7 & $\frac{1}{7} \cdot 2.8889$ & $ \frac{1}{7} \cdot 3.7320$ & $\frac{1}{7} \cdot 4$ & $\geq \frac{1}{7} \cdot 4.5000$   \\
\hline

$N>$ 6 & $ \frac{1}{N} \cdot \frac{26}{3^2} $  & $  \frac{1}{N}\left(\frac{4}{2}+\sqrt{3}\right)$  & $\frac{1}{N} \cdot 4$ & $\geq \frac{1}{N} \left(4+\frac{1}{2}\right) $  \\
\hline

\multicolumn{5}{|c|}{k = 4} \\
\hline

3 & $\frac{1}{3} \cdot 2.8749^*$ & $\frac{1}{3} \cdot 2.9943^*$ & $\frac{1}{3} \cdot 2.9943^*$ & $\geq \frac{1}{3} \cdot  2.9943$    \\
\hline

4 & $\frac{1}{4} \cdot 3.1180^*$ & $\frac{1}{4} \cdot 3.8649^*$ & $\frac{1}{4} \cdot 3.9664^*$ & $\geq \frac{1}{4} \cdot 3.9664$   \\
\hline

5 & $\frac{1}{5} \cdot 3.2188$ & $ \frac{1}{5} \cdot 4.4621$  & $ \frac{1}{5} \cdot 4.7860^*$ & $\geq \frac{1}{5} \cdot 4.7860$   \\
\hline

6 & $\frac{1}{6} \cdot 3.2188$ & $ \frac{1}{6} \cdot 4.4621$  & $ \frac{1}{6} \cdot 4.9494^*$ & $\geq \frac{1}{6} \cdot 5.4738$  \\
\hline

7 & $\frac{1}{7} \cdot 3.2188$ & $ \frac{1}{7} \cdot 4.4621$ & $ \frac{1}{7} \cdot 4.9665^*$ & $\geq \frac{1}{7} \cdot 5.6371$ \\
\hline

$N>$ 12 & $ \frac{1}{N} \cdot \frac{206}{4^3} $  & $  \frac{1}{N} \left(\frac{21+6\sqrt{6}}{8} \right)$  & $\frac{1}{N} \cdot 5$ & $\geq \frac{1}{N} \left(5+\frac{11}{16}\right)$\\
\hline

\multicolumn{5}{|c|}{k = 5} \\
\hline

3 & $\frac{1}{3} \cdot 2.9375^*$ & $\frac{1}{3} \cdot 2.9985^*$ & $\frac{1}{3} \cdot 2.9985^*$ & $\geq \frac{1}{3} \cdot 2.9985$ \\
\hline

4 & $\frac{1}{4} \cdot 3.3389^*$ & $\frac{1}{4} \cdot 3.9365^*$ & $\frac{1}{4} \cdot 3.9876^*$ & $\geq \frac{1}{4} \cdot 3.9876$  \\
\hline

5 & $\frac{1}{5} \cdot 3.4442^*$ & $\frac{1}{5} \cdot 4.6596^*$ & $\frac{1}{5} \cdot 4.9045^*$ & $\geq \frac{1}{5} \cdot 4.9045$  \\
\hline

6 & $\frac{1}{6} \cdot \frac{2194}{5^4}$ & $ \frac{1}{6} \cdot 5.1174$  & $ \frac{1}{6} \cdot 5.7750^*$ & $\geq \frac{1}{6} \cdot 5.7750$ \\
\hline

$N>$ 12 & $ \frac{1}{N} \cdot \frac{2194}{5^4} $  & $  \frac{1}{N} \left(\frac{8+5\sqrt{2}+\sqrt{5}+\sqrt{10}}{4}\right)$  & $\frac{1}{N} \cdot 6$ & $\geq \frac{1}{N} \left(6+\frac{26}{32}\right)$ \\
\hline

\multicolumn{5}{|c|}{$k \in \mathbb{N}$  } \\
\hline

large enough\footnote{For the classical and pure real case, we just need $N\geq k+1$. For the quantum pure and quantum case, we need $N$ to be large enough for a $k$-state design to exist, see Sec.~\ref{subsection:state_design} for details.} $N$ & $\frac{1}{N} \left( \frac{k!}{k^k}  \sum_{j=0}^k \frac{k^j}{j!}\right) $ & $\frac{1}{N} \left(\frac{1}{2^k} \left( \sum_{j=0}^k \sqrt{\binom{k}{j}} \right)^2\right)$ & $\frac{1}{N} (k+1)$ & $\geq \frac{1}{N}  \Big((k+1) +\left(1 - \frac{k+1}{2^k}\right)\Big) $ \\

  & $  \sim \frac{1}{N}   \sqrt{\frac{\pi k}{2}}$ & $ \sim \frac{1}{N} \sqrt{2\pi k}$ & $\sim \frac{1}{N} (k+1)$ & $ \geq \sim \frac{1}{N} (k+2)$ \\
\hline
\end{tabular}
\end{adjustbox}
\end{center}
\caption{The summary of the results for qubits, $d=2$, for different set cardinalities $N$ and numbers of copies $k$. We put an asterisk on the values for which we have a support of computational results or an educated guess, but no analytical proof, and $\geq$ near the values for which we have only the analytical lower bound.} \label{table:numerics}
\end{table*}
\clearpage
\section{Numerical results} \label{app:numerical}

We now present the numerical results for lower and upper bounds of $\Omega(d,N,k)$ organised in tables.

First, in \cref{table:classical_and_rebit}, we show numerical lower bounds for the maximal discriminability of sets of classical and real states. Here $\Omega^\textup{cl}_{\textup{uniform}}$ is for states distributed uniformly along the $Z$ axis, $\Omega^{\textup{cl}}$ is an optimisation over all classical distributions, and $\Omega^{\textup{cl}}_{\textup{UB}}$ is an analytical upper bound for classical states. Analogously, $\Omega^\textup{real}_{\textup{polygon}}$ is for regular polygon ensembles of pure rebits states, $\Omega^{\textup{real}}_{\textup{pure}}$ is an optimisation over pure rebit states, and $\Omega^{\textup{real}}$ is an optimisation over pure and mixed rebits states. The computational methods used are described in \cref{sec:numerics_lower}.

\begin{table}[h!]
\begin{center}
\begin{adjustbox}{max width=\linewidth}
\footnotesize
\setlength{\tabcolsep}{8pt}
\renewcommand{\arraystretch}{1.6}

\begin{tabular}{|c|c|c|c||c|c|c|} 
\hline
N & $\Omega^\textup{cl}_{\textup{uniform}}$ & $\Omega^\textup{cl}$ & $\Omega^\textup{cl}_\textup{UB}$ & $\Omega^\textup{real}_{\textup{polygon}}$ & $\Omega^\textup{real}_\textup{pure}$ & $\Omega^\textup{real}$\\ 
\hline

\multicolumn{7}{|c|}{k = 2} \\
\hline

3 & 0.8333 & 0.8333 & 0.8333 & 0.9714 & 0.9714 & 0.9714 \\
\hline

4 & 0.6111 & 0.625 & 0.625 & 0.7286 & 0.7286 & 0.7911 \\
\hline

5 & 0.5 & 0.5 & 0.5 & 0.5828 & 0.5828 & 0.6328 \\
\hline

6 & 0.4133 & 0.4167 & 0.4167 & 0.4857 & 0.4857 & 0.5274 \\
\hline

\multicolumn{7}{|c|}{k = 3} \\
\hline

3 & 0.9167 & 0.9167 & 0.963 & 0.9916 & 0.9916 & 0.9916 \\
\hline

4 & 0.7222 & 0.7222 & 0.7222 & 0.9330 & 0.9330 & 0.9330 \\
\hline

5 & 0.5687 & 0.5778 & 0.5778 & 0.7464 & 0.7464 & 0.8464 \\
\hline

6 & 0.4773 & 0.4815 & 0.4815 & 0.622 & 0.622 & 0.7053 \\
\hline

\multicolumn{7}{|c|}{k = 4} \\
\hline

3 & 0.9583 & 0.9583 & 1 & 0.9981 & 0.9981 & 0.9981 \\
\hline

4 & 0.7716 & 0.7795 & 0.8047 & 0.9662 & 0.9662 & 0.9662 \\
\hline

5 & 0.6438 & 0.6438 & 0.6438 & 0.8924 & 0.8924 & 0.9230 \\
\hline

6 & 0.5275 & 0.5365 & 0.5365 & 0.7437 & 0.7437 & 0.8583 \\
\hline

\multicolumn{7}{|c|}{k = 5} \\
\hline

3 & 0.9792 & 0.9792 & 1 & 0.9995 & 0.9995 & 0.9995 \\
\hline

4 & 0.8292 & 0.8347 & 0.8776 & 0.9841 & 0.9841 & 0.9841 \\
\hline

5 & 0.6832 & 0.6888 & 0.7021 & 0.9319 & 0.9319 & 0.9623 \\
\hline

6 & 0.5851 & 0.5851 & 0.5851 & 0.8529 & 0.8529 & 0.9172 \\
\hline

\multicolumn{7}{|c|}{k = 6} \\
\hline

3 & 0.9896 & 0.9896 & 1 & 0.9999 & 0.9999 & 0.9999 \\
\hline

4 & 0.8512 & 0.8668 & 0.9437 & 0.9920 & 0.9920 & 0.9920 \\
\hline

5 & 0.7235 & 0.7288 & 0.7549 & 0.9581 & 0.9581 & 0.9811 \\
\hline

6 & 0.6142 & 0.6203 & 0.6291 & 0.8943 & 0.8943 & 0.9520  \\
\hline
 
\end{tabular}
\end{adjustbox}
\end{center}
\caption{Computational lower bounds for classical and real quantum sets of states for $d=2$ and different $N$, $k$.}
\label{table:classical_and_rebit}
\end{table}

In \cref{table:upper_numerics_sdp} we present results for upper bounds for maximal discriminability for different sets of states with dimension $d=2$:
$\Omega^{\textup{real}}_{\textup{pure}}$ is an optimisation over pure rebit states, $\Omega^{\textup{real}}$ is an optimisation over pure and mixed rebits states, $\Omega^{\textup{q}}_{\textup{pure}}$ - over pure quantum states, and $\Omega^{\textup{q}}$ - over all pure and mixed quantum states.

\begin{table}

\begin{center}
\begin{adjustbox}{max width=\linewidth}
\footnotesize
\setlength{\tabcolsep}{8pt}
\renewcommand{\arraystretch}{1.6}
\begin{tabular}{|c|c|c|c|c|}
\hline
N& $\Omega^{\textup{real}}_{\textup{pure}}$& $\Omega^{\text{real}}$& $\Omega^{\textup{q}}_{\textup{pure}}$& $\Omega^{\textup{q}}$ \\
\hline

\multicolumn{5}{|c|}{k = 2} \\
\hline
3 & $= 0.9714$ & $\le 0.9714$ & $\le 1.0000$ & $\le 1.0000$ \\
\hline

4 & $= 0.7286$ & $\le 0.8067$ & $\le 0.7500$ & $\le 0.8333$ \\
\hline

5 & $= 0.5828$ & $\le 0.6453$ & $\le 0.6000$ & $\le 0.6667$ \\
\hline

6 & $= 0.4857$ & $\le 0.5378$ & $\le 0.5000$ & $\le 0.5556$ \\
\hline

\multicolumn{5}{|c|}{k = 3} \\
\hline
3 & $= 0.9916$ & $\le 0.9921$ & $\le 1.0000$ & $\le 1.0000$ \\
\hline

4 & $= 0.9330$ & $\le 0.9378$ & $\le 1.0000$ & $\le 1.0000$ \\
\hline

5 & $= 0.7464$ & $\le 0.8740$ & $\le 0.8000$ & $\le 0.9367$ \\
\hline

6 & $= 0.6220$ & $\le 0.7643$ & $\le 0.6667$ & $\le 0.8177$ \\
\hline

\multicolumn{5}{|c|}{k = 4} \\
\hline
3 & $= 0.9981$ & $\le 1.0000$ & $\le 1.0000$ & $\le 1.0000$ \\
\hline

4 & $= 0.9662$ & $\le 1.0000$ & $\le 1.0000$ & $\le 1.0000$ \\
\hline

5 & $= 0.8924$ & $\le 1.0000$ & $\le 1.0000$ & $\le 1.0000$ \\
\hline

6 & $= 0.7437$ & $\le 1.0000$ & $\le 0.8333$ & $\le 1.0000$ \\

\hline

\end{tabular}
\end{adjustbox}
\end{center}

\caption{Upper bounds on $\Omega^{\textup{q}}(d=2,N,k)$ for various scenarios, obtained through the SDP described in \cref{sec:numerics_upper}. The number of extensions $\ell$ was restricted to $\ell \le 5 - k$ due to memory constraints. We notice that the SDP does not converge to a non-trivial bound in several scenarios, including the $k = 2$ and $N = 3$ quantum case. For $\Omega^{\textup{real}}_{\textup{pure}}$, all values matched the lower bounds in \cref{table:classical_and_rebit} up to numerical precision, indicating the exact value has been found.}
\label{table:upper_numerics_sdp}
\end{table}

\clearpage
\section{Proof of \cref{thm:SO(2)}} \label{app:SO(2)}

\thmRebit*
\begin{proof}
    When $d=2$, a real state $k$-design exists for every $N\geq k+1$, it is enough to consider a set of states distributed in the ZX-plane as a regular polygon~\cite{Delsarte1991spherical}. Hence, for any $N\geq k+1$, we can apply the lower bound presented in \cref{thm:SO(d)} for the particular case of $d=2$. We now notice that, when $d=2$, the multiplicity 
\begin{align}
     m_j = \binom{d + k - 2j - 1}{d - 1} - \binom{d + k - 2j - 3}{d - 1}
\end{align}
simplifies to 
\begin{align}
m_j = 
\begin{cases} 
2, & \text{if } j < \frac{k}{2} \\
1, & \text{if } j = \frac{k}{2}
\end{cases}
\end{align}
(notice that  $m_j=1$ only happens when $k$ is even and only for the greatest allowed value $j$).

Let us now analyse the value of $\lambda_j$ when $d=2$. For arbitrary $d$, we have that,
\begin{equation}
\lambda_j = \frac{k!(d - 2)!!}{(2j)!!(d + 2k - 2j - 2)!!}
\end{equation}
For $d=2$, the term $(d-2)!!$ becomes $0!! = 1$. The denominator term $(d + 2k - 2j - 2)!!$ simplifies to $(2 + 2k - 2j - 2)!! = (2k - 2j)!!$:
\begin{equation}
\lambda_j = \frac{k!}{(2j)!! (2k - 2j)!!}
\end{equation}

The mathematical identity $(2n)!! = 2^n n!$ ensures that the terms in the denominator can be re-written as $(2j)!! = 2^j j!$ and $(2k - 2j)!! = 2^{k-j} (k - j)!$.
Substituting these into the expression, we get
\begin{equation}
\lambda_j = \frac{k!}{(2^j j!) (2^{k-j} (k - j)!)}
\end{equation}

Combining the powers of 2 in the denominator, we obtain
\begin{equation}
\lambda_j = \frac{k!}{2^k j! (k - j)!}
\end{equation}
We can then recognise the binomial coefficient $\binom{k}{j} = \frac{k!}{j!(k-j)!}$ to write
\begin{equation}
\lambda_j = \frac{1}{2^k} \binom{k}{j}.
\end{equation}

We now want to evaluate the quantity $p = S^2$ where $S = \sum_{j=0}^{\lfloor k/2 \rfloor} m_j \sqrt{\lambda_j}$. 

Substituting $\lambda_j$ into the sum gives:
\begin{align}
S &= \sum_{j=0}^{\lfloor k/2 \rfloor} m_j \sqrt{\frac{1}{2^k} \binom{k}{j}} \\
  &= \sqrt{\frac{1}{2^k} } \sum_{j=0}^{\lfloor k/2 \rfloor} m_j \sqrt{\binom{k}{j}}
\end{align}

We now assume that $k$ is even. Given $\lambda_j = 2^{-k}\binom{k}{j}$, we know that for $j < k/2$, the multiplicity is $m_j = 2$, and when $j = k/2$, the multiplicity is $m_{k/2} = 1$. 
We can write the sum by explicitly separating the middle term from the rest:
\begin{align}
S &= \left( \sum_{j=0}^{k/2 - 1} m_j \sqrt{\lambda_j} \right) + m_{k/2} \sqrt{\lambda_{k/2}} \\
  &= \left( \sum_{j=0}^{k/2 - 1} 2 \sqrt{2^{-k}\binom{k}{j}} \right) + 1 \sqrt{2^{-k}\binom{k}{k/2}}
\end{align}

Factoring out the common $\sqrt{\frac{1}{2^k}}$ term yields:
\begin{align}
S &= \sqrt{\frac{1}{2^k}} \left[ \sum_{j=0}^{k/2 - 1} 2 \sqrt{\binom{k}{j}} + \sqrt{\binom{k}{k/2}} \right]
\end{align}

Using the symmetry of binomial coefficients, $\binom{k}{j} = \binom{k}{k-j}$, we can split the factor of 2 in the summation to represent the ``left'' and ``right'' halves of the distribution:
\begin{align}
\sum_{j=0}^{k/2 - 1} 2 \sqrt{\binom{k}{j}} &= \sum_{j=0}^{k/2 - 1} \sqrt{\binom{k}{j}} + \sum_{j=0}^{k/2 - 1} \sqrt{\binom{k}{k-j}} \\
&= \sum_{j=0}^{k/2 - 1} \sqrt{\binom{k}{j}} + \sum_{i=k/2 + 1}^{k} \sqrt{\binom{k}{i}}
\end{align}

Substituting this back into our equation for $S$, the pieces seamlessly join together to form a complete sum from $0$ to $k$:
\begin{align}
S &= \sqrt{\frac{1}{2^k}} \left[ \sum_{j=0}^{k/2 - 1} \sqrt{\binom{k}{j}} + \sqrt{\binom{k}{k/2}} + \sum_{i=k/2 + 1}^{k} \sqrt{\binom{k}{i}} \right] \\
  &= \sqrt{\frac{1}{2^k}} \sum_{j=0}^{k} \sqrt{\binom{k}{j}}
\end{align}
Squaring the entire sum gives the final quantity for $p$:
\begin{align}
p =  \frac{1}{2^k} \left( \sum_{j=0}^{k} \sqrt{\binom{k}{j}} \right)^2.
\end{align}

The calculation for the case of odd $k$ is analogous.
\end{proof}

\section{Proof of \cref{thm:trine_new}} \label{app:trine}

\thmTrineNew*
\begin{proof}
    Since the set of real pure qubit states is a subset of all pure qubit states, it follows directly from \cref{thm:SO(2)} that $\Omega^{\textup{q}}_\textup{pure}(2,3,2) \geq \Omega^\textup{real}_\textup{pure}(2,3,2) \geq \frac{1}{2}+\frac{\sqrt{2}}{3}$. Hence, it remains to be shown that it is also an upper bound of $\Omega^{\textup{q}}_\textup{pure}(2,3,2)$, i.e., $\Omega^{\textup{q}}_\textup{pure}(2,3,2) \leq \frac{1}{2}+\frac{\sqrt{2}}{3}$. From there, Eq.~\eqref{eq:232_value} would follow directly.

    We will start with the proof for real states (rebits), which will then be generalised for general pure qubit states. The idea is to construct a feasible point for the dual problem of the state discrimination task that is independent of the states that are subject to the discrimination. By weak duality~\cite{Boyd2004}, this dual feasible point gives rise to an upper bound of the corresponding primal problem. 
    In general, the dual problem of quantum state discrimination for a set of states $\{\sigma_l\}_{l=1}^N$ which is uniformly distributed is given by~\cite{Bae2015ReviewQSD}
    \begin{align} \label{eq:dual}
       & \min \Tr(G) \\
       & \textup{s.t. }\frac{1}{N} \sigma_l \leq G, \quad \forall l \in \{1,\ldots, N\}.         
    \end{align}
    Let us now focus on the case where $\sigma_l = \ketbra{\psi_l} \otimes \ketbra{\psi_l}$, and where $\ket{\psi_l} \in \mathbb{R}^2$. Since  $\ket{\psi_l}$ is a rebit, we can generate it as 
    \begin{align}
         \ket{\psi_l} = O_l \ket{0},
    \end{align}
    where $O_l\in \textup{SO}(2)$ is real qubit unitary operator, that is, a rotation in the XZ-plane of the Bloch sphere, or concretely $O_l = e^{i \gamma_{l} \sigma_{y}}$ for some angle $\gamma_l \in [0, \pi]$. Hence, the problem reads as  
        \begin{align}
            & \min \Tr(G) \\
            & \textup{s.t. }\frac{1}{N} \Big(O_l\otimes O_l\Big) \ketbra{00} \Big(O_l \otimes O_l\Big)^\dagger\leq G, \quad \forall l \in \{1,\ldots, N\},       
        \end{align}
        or, equivalently,
        \begin{align}
         & \min \Tr(G) \\
         & \textup{s.t. }\frac{1}{N}\ketbra{00} \leq \Big(O_l\otimes O_l\Big)^\dagger G  \Big(O_l \otimes O_l\Big), \quad \forall i \in \{1,\ldots, N\},         
        \end{align}    
    for some given rotations $O_i \in \textup{SO}(2)$. The idea is to construct an ansatz for $G$ that is invariant under all rotations in the sense that 
    \begin{equation}
        \Big(O\otimes O\Big)^\dagger G  \Big(O \otimes O\Big) = G \; \forall O \in \textup{SO}(2).
    \end{equation}
    This is achieved by 
    \begin{align}
            G=\frac{1}{N} \left(\alpha \Pi_\textup{sym}(2,2) + \beta \ketbra{\phi^+}\right),
    \end{align}
    where $\ket{\phi^+} = \frac{1}{\sqrt{2}}(\ket{00} + \ket{11})$, due to the well-known facts that $\Big(U \otimes U\Big) \Pi_\textup{sym}(2,2)   \Big(U \otimes U\Big)^\dagger = \Pi_\textup{sym}(2,2)$ for all unitaries $U \in U(2)$ and $O \otimes O \ket{\phi^{+}} = \ket{\phi^{+}}$ for all $O \in \textup{SO}(2)$. Now, using this ansatz for $G$, the dual problem reduces to
        \begin{align}
            & \min \frac{1}{N}\left(3\alpha + \beta\right) \\
            & \textup{s.t. }\ketbra{00} \leq \alpha \Pi_\textup{sym}(2,2) + \beta \ketbra{\phi^+}.  \label{eq:xz_op_inequality}       
        \end{align}
       The operator inequality \eqref{eq:xz_op_inequality}  holds iff
    \begin{align}
        \alpha \Pi_\textup{sym}(2,2) + \beta \ketbra{\phi^+}- \ketbra{00}\geq0.    
    \end{align}
    This is equivalent to the condition that all eigenvalues of $M:=\alpha \Pi_\textup{sym}(2,2) + \beta \ketbra{\phi^+}- \ketbra{00}$ are non-negative.
    In the computational basis, we are hence interested in the eigenvalues of the matrix
    \begin{align}
        M = \begin{pmatrix} 
    \alpha + \frac{\beta}{2} - 1 & 0 & 0 & \frac{\beta}{2} \\ 
    0 & \frac{\alpha}{2} & \frac{\alpha}{2} & 0 \\ 
    0 & \frac{\alpha}{2} & \frac{\alpha}{2} & 0 \\ 
    \frac{\beta}{2} & 0 & 0 & \alpha + \frac{\beta}{2} 
    \end{pmatrix}.
    \end{align}
    The matrix $M$ is block diagonal, so we just need to calculate the eigenvalues of two-by-two matrices.
    Direct calculation shows that the eigenvalues are given by
    \begin{align}
        & \lambda_1 = 0 \\ 
        & \lambda_2 = \alpha \\ 
        & \lambda_3 = \alpha + \frac{\beta - 1 + \sqrt{\beta^2 + 1}}{2} \\ 
        & \lambda_4 = \alpha + \frac{\beta - 1 - \sqrt{\beta^2 + 1}}{2}. 
    \end{align}
    Hence, we must have $\alpha\geq0$ and since $\alpha,\beta\in\mathbb{R}$, the only other constraint that is relevant is $\lambda_4\geq 0 $, which imposes that
    \begin{align}
    \alpha \geq \frac{-\beta + 1 +  \sqrt{\beta^2 + 1}}{2}.
    \end{align}
    Since we want to minimise $3\alpha+\beta$, it makes sense to assume that
    \begin{align}
    \alpha = \frac{-\beta + 1 +  \sqrt{\beta^2 + 1}}{2}.
    \end{align}
    To then minimise $\Tr(G)$, our problem reads as 
    \begin{align}
       \frac{1}{N} \min_{\beta\in\mathbb{R}} \; 3\left( \frac{-\beta + 1 +  \sqrt{\beta^2 + 1}}{2}\right) + \beta.
    \end{align}
    In order to find the minimum of a real-valued function we just need to derive and find its extremal points of $f(\beta):= 3\left( \frac{-\beta + 1 +  \sqrt{\beta^2 + 1}}{2}\right) + \beta$.
    Direct calculation shows that the derivative of $f$ is given by
    \begin{align}
        f'(\beta)= \frac{3\beta}{2\sqrt{\beta^2 + 1}} - \frac{1}{2}.
    \end{align}
    We then see that $f'(\beta)=0$ when $\beta=\frac{\sqrt{2}}{4}$, and by substituting $\beta=\frac{\sqrt{2}}{4}$ into $f$, we obtain
    \begin{align}
        f\left(\frac{\sqrt{2}}{4}\right) =\frac{3}{2}+\sqrt{2}.
    \end{align}
    For completeness, let us explicitly check that $\alpha\geq0$, so that the eigenvalues of $M = G-\ketbra{00}$ are non-negative. Since we took $\alpha := \frac{-\beta + 1 +  \sqrt{\beta^2 + 1}}{2}$, when  $\beta=\frac{\sqrt{2}}{4}$, we have that
    \begin{align}
        \alpha=\frac{2+\sqrt{2}}{4} > 0.
    \end{align}
    
    To conclude, we have shown that, for $\alpha=\frac{2+\sqrt{2}}{4}$, $\beta=\frac{\sqrt{2}}{4}$, and $G=\frac{1}{N} \left(\alpha \Pi_\textup{sym}(2,2) + \beta \phi^+\right)$, for any pure rebit $\ket{\psi}\in\mathbb{R}^2$,
    it holds that
    \begin{align}
        \frac{1}{N}  \ketbra{\psi}\otimes\ketbra{\psi} \leq G. 
    \end{align}
    Hence, $\Tr(G)=\frac{1}{N}\left(\frac{3}{2}+\sqrt{2}\right)$ is an upper bound to the problem of pure rebit discrimination, i.e, $\Omega^\textup{real}_\textup{pure}(2,N,2) \leq \frac{1}{N}\left(\frac{3}{2}+\sqrt{2}\right)$ and by Theorem \ref{thm:SO(2)} $\Omega^\textup{real}_\textup{pure}(2,N,2) = \frac{1}{N}\left(\frac{3}{2}+\sqrt{2}\right)$. In particular, for $N = 3$ we have $\Omega^\textup{real}_\textup{pure}(2,3,2)= \frac{1}{2}+\frac{\sqrt{2}}{3} \approx 0.9714$
    
    The proof logic for the real case can be generalized to complex qubit states if $N = 3$.
    To see this, the important insight is that three pure qubit states $\ket{\psi_i}$ always lie in some common plane that intersects the Bloch sphere. By choosing an appropriate coordinate system (performing a unitary rotation) this plane is without loss of generality parallel to the XZ-plane meaning that $ \ket{\psi_i} \in \mathcal{O}_{\theta} := \{\ket{\psi} \in \mathbb{C}^2: \braket{\psi} = 1, \bra{\psi} \sigma_y \ket{\psi} = -\sin(2\theta)\}$ for some constant angle $\theta \in [-\pi/4,\pi/4]$. All the pure states in this plane can be transformed into each other by rotations $O \in \textup{SO}(2)$ and $\textup{SO}(2)$ leaves each of the sets $\mathcal{O}_{\theta}$ invariant. Starting from the state $\ket{0}$, which lies in the rebit plane $\mathcal{O}_{0}$, it can be simply checked that the state $e^{-i\theta \sigma_{x}} \ket{0}$ lies in $\mathcal{O}_{\theta}$. 
        
    Suppose now that $G(\theta)$ is a two-qubit operator that is invariant under the unitary conjugation by $O \otimes O$ for all $O \in \textup{SO}(2)$, i.e.,
    \begin{equation}
        \Big(O\otimes O\Big)^\dagger G(\theta)  \Big(O \otimes O\Big) = G(\theta) \; \forall O \in \textup{SO}(2).
    \end{equation}
    Then, by the same logic as in the real case, one has the implication 
    \begin{equation}
        \label{eq:y-rotation_implication}
        \left[G(\theta) \geq \frac{1}{3} (e^{-i\theta\sigma_x} \otimes e^{-i\theta \sigma_x})\ketbra{00} (e^{i\theta\sigma_x} \otimes e^{i\theta \sigma_x}) \right]  \Rightarrow \left[G(\theta) \geq \frac{1}{3} \ketbra{\psi} \otimes \ketbra{\psi} \; \forall \ket{\psi} \in \mathcal{O}_{\theta} \right]
    \end{equation}
    and hence, $\Tr(G(\theta))$ is guaranteed be an upper bound to the success probability in uniform state discrimination of two copies of three states from $\mathcal{O}_{\theta}$, whenever the left condition of the implication \eqref{eq:y-rotation_implication} is satisfied.
    As an ansatz for the operator $G(\theta)$, we assume that it can be decomposed as 
    \begin{equation}
        G(\theta) = \sum_{i=1}^{3} \lambda_i(\theta) \ketbra{\alpha_i}
    \end{equation}
    with the orthonormal eigenvectors $\ket{\alpha_i}$ that are in terms of the Bell states given by\footnote{This ansatz was motivated by numerically solving the SDP of the dual problem presented in \cref{eq:dual} and analysing the structure of the optimal operator $G$. }
    \begin{align}
        \ket{\alpha_1} = \ket{\phi^+}, \quad \ket{\alpha_2} = \frac{1}{\sqrt{2}}(\ket{\phi^-} - i \ket{\psi^+}), \quad  \ket{\alpha_3} = \frac{1}{\sqrt{2}}(\ket{\phi^-} + i \ket{\psi^+})
    \end{align}
    independent of $\theta$. The corresponding eigenvalues $\lambda_i(\theta)$ will depend on $\theta$ and are constructed in the following. Before that, it should be noticed that all the vectors $\ket{\alpha_i}$ are eigenvectors of $O \otimes O$ for all $O \in \textup{SO}(2)$
    so that $G(\theta)$ is invariant under conjugation by  $O \otimes O$ and hence, it suffices for $G(\theta)$ to obey 
    \begin{equation}
        \label{eq:rotated_dual_inequality}
       G(\theta) \geq \frac{1}{3} (e^{-i\theta\sigma_x} \otimes e^{-i\theta \sigma_x})\ketbra{00} (e^{i\theta\sigma_x} \otimes e^{i\theta \sigma_x}).
    \end{equation}
    We denote the eigenvalues $\lambda_i(\theta)$ that satisfy the inequality \eqref{eq:rotated_dual_inequality} as feasible. To understand the condition under which the eigenvalues are feasible, notice first that one can decompose the state $(e^{-i\theta\sigma_x} \otimes e^{-i\theta \sigma_x})\ket{00}$ in the eigenbasis of $G(\theta)$ via 
    \begin{equation}
        \label{eq:basis_decomposition}
        (e^{-i\theta\sigma_x} \otimes e^{-i\theta \sigma_x})\ket{00} = \frac{\cos(2\theta)}{\sqrt{2}} \ket{\alpha_1} + \left(\frac{1+\sin(2\theta)}{2} \right) \ket{\alpha_2} + \left(\frac{1-\sin(2\theta)}{2}\right) \ket{\alpha_3}.
    \end{equation}
    Next, we notice that the operator inequality \eqref{eq:rotated_dual_inequality} is equivalent to the scalar inequality 
    \begin{equation}
        \label{eq:scalar_inequality}
       \frac{1}{3} \bra{00} (e^{i\theta\sigma_x} \otimes e^{i\theta \sigma_x}) G(\theta)^{-1} (e^{-i\theta\sigma_x} \otimes e^{-i\theta \sigma_x}) \ket{00} \leq 1 
    \end{equation}
    where $G(\theta)^{-1}$ denotes the inverse on the support of $G(\theta)$. To see this, let $A$ be a positive definite matrix, $\ket{a}$ be any normalised vector and consider the following logic of equivalences:
    \begin{align}
        A \geq \ketbra{a} &\Leftrightarrow A - \ketbra{a} \geq 0 \\
        & \Leftrightarrow A^{\frac{1}{2}}(\mathds{1} - A^{-\frac{1}{2}} \ketbra{a} A^{-\frac{1}{2}})  A^{\frac{1}{2}} \geq 0 \\
        & \Leftrightarrow \mathds{1} - A^{-\frac{1}{2}} \ketbra{a} A^{-\frac{1}{2}} \geq 0 \\
        & \Leftrightarrow \bra{a} A^{-1} \ket{a} \leq 1
    \end{align}
    Thus, using the inequality \eqref{eq:scalar_inequality} and the basis decomposition \eqref{eq:basis_decomposition}, the eigenvalues $\lambda_i(\theta)$ are feasible if and only if they obey the inequality 
    \begin{equation}
        \label{eq:eigenvalues_dual_feasibility}
        \frac{1}{3} \left(\frac{\cos^2(2 \theta)}{2 \lambda_1(\theta)} +  \frac{(1 + \sin(2\theta))^2}{4 \lambda_2(\theta)} + \frac{(1 - \sin(2\theta))^2}{4 \lambda_3(\theta)}\right) \leq 1.
    \end{equation}
    
    One may now check that one set of feasible eigenvalues is given by $\lambda_i(\theta) = \sqrt{a_i(\theta)} \sum_{j=1}^3  \sqrt{a_j(\theta)}$, where $a_i(\theta)$ are the coefficients in the inequality \eqref{eq:eigenvalues_dual_feasibility}, i.e.,
    \begin{align}
        a_1(\theta) &= \frac{\cos^2(2 \theta)}{6}, \quad a_2(\theta) = \frac{(1 + \sin(2\theta))^2}{12}, \quad a_3(\theta) = \frac{(1 - \sin(2\theta))^2}{12}.
    \end{align}
    
    With these eigenvalues $\lambda_i(\theta)$, one computes
    \begin{align}
        \Tr(G(\theta)) &= \sum_{i=1}^{3} \lambda_i(\theta) = \left(\sum_{i=1}^{3} \sqrt{a_i(\theta)}\right)^2 \\
        &= \left(\frac{\cos(2\theta)}{\sqrt{6}} + \frac{1+\sin(2 \theta)}{\sqrt{12}} + \frac{1-\sin(2 \theta)}{\sqrt{12}}\right)^2 \\
        &= \left(\frac{2}{\sqrt{12}} + \frac{\cos(2 \theta)}{\sqrt{6}}\right)^2 \leq \left(\frac{2}{\sqrt{12}} + \frac{1}{\sqrt{6}}\right)^2 = \frac{1}{2}+\frac{\sqrt{2}}{3},
    \end{align}
    which shows that $\Omega^{\textup{q}}_\textup{pure}(d=2,N=3,k=2) \leq \frac{1}{2}+\frac{\sqrt{2}}{3}$ concluding the proof.
\end{proof}

\section{Implementation details of the gradient descent for lower bounds}\label{app:gd_details}

In order to perform the gradient descent optimisation, we first modify the problem to be unconstrained over its parameter space. To achieve this, instead of optimizing over states $\rho_i$ and POVM elements $M_i$ directly, we optimize over unconstrained matrices $A_i \in \mathbb{C}^{d \times d}$ and $B_i \in \mathbb{C}^{d^k \times d^k}$, respectively, which we use to construct the valid operators. To obtain valid quantum states, we require for all $i$
\begin{align}
    \rho_i \geq 0 \quad \qquad \text{and} \qquad \Tr(\rho_i) = 1.
\end{align}
The positivity and unit trace constraints can be fulfilled by setting $\rho_i = ( A_i^\dagger A_i ) / \operatorname{Tr}(A_i^\dagger A_i)$. For the POVM $\{M_i\}_{i=1}^{N}$, it should satisfy
\begin{equation}
    M_i \geq 0 \quad \forall i, \qquad \text{and} \qquad \sum_{i=1}^{N} M_i = \id.
\end{equation}
Positivity can be achieved by setting $\tilde{M}_i = B_i^\dagger B_i$. Letting $M_i = \tilde{M}_i / \lambda_\text{max}(\sum_{i=1}^N \tilde{M}_i)$, where $\lambda_\text{max}(Q)$ is the maximum eigenvalue of $Q$, we obtain $\sum_{i=1}^N M_i \le \id$. This is sufficient for the gradient descent optimisation, as the gradient will naturally favour solutions where $\sum_{i=0}^N M_i = \id$ is satisfied. Indeed, this relaxation provides a less constrained version of the search space to be explored in the initial iterations, which is often beneficial. For pure states, we optimize the vectors $\ket{a_i}$ such that $A_i = \ketbra{a_i}{a_i}$, and for rebits, we restrict the matrix $A_i$ or vectors $\ket{a_i}$ to be real.

\section{Implementation details of the SDP for upper bounds}\label{app:sdp_details}

In practice, we have implemented the SDP in \cref{eq:sdp_main}, for $d = 2$, by considering a permutation-invariant basis in terms of the Pauli operators. First, we define a (non-normalised) basis for $\operatorname{Perm}_d^k$ via
\begin{equation}
	b_k[t] = \sum_{\mu \in \mathcal{S}_k} V_\mu ( \id^{\otimes t_\id} \otimes X^{\otimes t_X} \otimes Y^{\otimes t_Y} \otimes Z^{\otimes t_Z} ) V^\dagger_\mu
\end{equation}
where $t$ denotes the type~\cite{harrow2013churchsymmetricsubspace} indexing the basis, i.e., a 4-tuple of non-negative integers such that $\sum_{\alpha \in \{\id, X, Y, Z\}} t_\alpha = k$. Next, we define the joint basis for states and measurements via
\begin{equation}
	B[t,t'] = \frac{1}{d^{k+\ell}} b_{k+\ell}[t] \otimes b_{k}[t']
\end{equation}
such that the real coefficients $\phi_i[t,t']$ define $\Phi^\ell_i = \sum_{t, t'} \phi_i[t,t'] B[t,t']$, and similarly for $\Phi^\ell_{\perp i}$. The Pauli representation allows the constraints to be written directly in terms of elementary operations on these types. For example, the purity constraint can be imposed first by defining $r_i[a,b,c] = \frac{1}{d^k}\left( \phi_i[(k+l-2,a,b,c),(k,0,0,0)] + \phi_{\perp i}[(k+l-2,a,b,c),(k,0,0,0)] \right)$, for $a + b + c = 2$, corresponding to the  coefficients associated with the quadratic Pauli terms in the two-copy reduced state. The purity condition then reduces to $r[2,0,0] + r[0,2,0] + r[0,0,2] = 1$, which is equivalent to imposing a unit Bloch vector for the state.

Ultimately, for the positive semidefinite constraints we opted to construct the $2^{2k+l} \times 2^{2k+l}$ matrices directly using this permutation-invariant Pauli basis. Partial transposes were implemented using a direct transposition of the basis elements $B[t,t']$. This is inefficient and renders the SDP optimisation intractable for larger $k$ and $\ell$, but it was sufficient as a proof of concept. We note that many further optimisations are possible, such as exploiting Schur-Weyl duality, where $\Phi^k_i, \Phi^k_{\perp i}, \Psi_i$, and their partial transposes, can be written in block diagonal form. Positivity can then be imposed directly on their unique blocks, each corresponding to an irreducible representation of the symmetric group. However, we did not pursue such optimisations here. For further details, see the comments in the code available on GitHub~\cite{1ucasvb_github}.

\end{document}